\documentclass[journal]{IEEEtran}
\usepackage[english]{babel}

\usepackage{algorithmic}
\usepackage[ruled,vlined]{algorithm2e}
\usepackage{amsfonts}
\usepackage{amsmath}
\usepackage{amsthm}
\usepackage{amssymb}
\usepackage{bbm}
\usepackage{cite}
\usepackage{color}
\usepackage[mathscr]{euscript}
\usepackage{enumitem}
\usepackage{float}
\usepackage[T1]{fontenc}
\usepackage{glossaries}
\usepackage{graphicx}
\usepackage[latin9]{inputenc}
\usepackage{times}
\usepackage{latexsym}
\usepackage{mathtools}
\usepackage{multirow}
\usepackage{makeidx}
\usepackage{mathrsfs} 
\usepackage{nomencl}
\usepackage{rotating}
\usepackage{textcomp}
\usepackage{varioref}
\usepackage{xcolor}

\ifCLASSOPTIONcompsoc
\usepackage[caption=false,font=normalsize,labelfon
t=sf,textfont=sf]{subfig}
\else
\usepackage[caption=false,font=footnotesize]{subfi
g}
\fi

\DeclareGraphicsExtensions{.eps,.png}

\usepackage{soul}
\usepackage[scaled=.80]{beramono}
\usepackage{hyperref}
\hypersetup{
    colorlinks=true,
    linkcolor=black,
    filecolor=black,      
    urlcolor=black,
    citecolor = blue,
}

\onecolumn
\usepackage{cite,setspace,booktabs}
\usepackage[switch]{lineno}

\ifCLASSINFOpdf
\else
\fi

\begin{document}
\bstctlcite{IEEEexample:BSTcontrol}
\renewcommand{\thefootnote}{\normalsize \arabic{footnote}} 	
\newtheorem{theo}{Theorem}
\newtheorem{conj}{Conjecture}
\newtheorem{definition}{Definition}
\newtheorem{prop}{Proposition}
\newtheorem{lem}{Lemma}
\newtheorem{cor}{Corollary}
\newtheorem{remark}{Remark}
\newtheorem{example}{Example}

\newcommand{\pw}{PW_{\Omega}}
\newcommand{\kz}{k\in\mathbb{Z}}
\newcommand{\akz}{\forall k \in \mathbb{Z}}
\newcommand{\tu}{\mathcal{T}_u}
\newcommand{\db}{\bar{\delta}}
\newcommand{\Z}{\mathbb{Z}}
\newcommand{\N}{\mathbb{N}}
\newcommand{\R}{\mathbb{R}}
\newcommand{\LR}{L^2(\mathbb{R})}
\newcommand{\LO}{L^2([-\Omega,\Omega])}
\newcommand{\re}{\mathbb{R}}
\newcommand{\co}{\mathbb{C}}
\newcommand{\cH}{\mathcal{H}}
\newcommand{\Tu}{\mathcal{T}_u}
\newcommand{\F}{\mathcal{F}}
\newcommand{\bs}{\boldsymbol}
\newcommand{\eq}{\triangleq}
\newcommand{\subSz}{\mathcal{K}_N}
\newcommand{\Sz}{\mathbb{S}_N}
\newcommand{\s}{ISI}
\newcommand{\MO}[0]{\mathscr{M}_\lambda}
\newcommand{\fig}[1]{Fig.~\ref{#1}}
\newcommand{\MOh}[0]{\mathscr{M}_{\boldsymbol{\mathsf{H}}}}

\newcommand{\fe}[1]{\left[\kern-0.30em\left[#1  \right]\kern-0.30em\right]}

\newcommand{\vb}[1]{\left\lvert #1 \right\rvert}
\newcommand{\rb}[1]{\left( #1 \right)}
\newcommand{\sqb}[1]{\left[ #1 \right]}
\newcommand{\cb}[1]{\left\lbrace #1 \right\rbrace}
\newcommand{\floor}[1]{\left\lfloor #1 \right\rfloor}
\newcommand{\ceil}[1]{\left\lceil #1 \right\rceil}

\newcommand\ab[1]			{{\color{red}#1}}
\newcommand\bb[1]			{{\color{blue}#1}}
\newcommand\ETP[1]         {\mathrm{E}_{p}\left(#1\right)}

\newcommand\ETPP[2]         {\mathrm{E}_{#1}\left(#2\right)}

\renewcommand\tilde{\widetilde}

\DeclarePairedDelimiter{\norm}{\Vert}{\Vert}
\DeclarePairedDelimiter{\abs}{\left|}{\right|}
\DeclarePairedDelimiter{\Prod}{\langle}{\rangle}

\newcommand{\PW}[1]{\mathsf{PW}_{#1}}
\newcommand{\iPW}[2]{#1 \in \mathsf{PW}_{#2}}

\renewcommand\geq\geqslant \renewcommand\leq\leqslant
\newcommand{\const}[1]{#1}

\def\th{\theta}

\def\thbar{\theta_\beta}

\def\ind{{\color{red}{\pmb{1}}}}

\def\ind{\mathbbmtt{1}}

\def\figmode    {1}         
\def\PH         {0}         
\def\stabilization	{0}
\def\True		{1}

\title{The Surprising Benefits of Hysteresis in Unlimited Sampling: Theory, Algorithms and Experiments}

\author{Dorian~Florescu, 
	Felix~Krahmer
	and~Ayush~Bhandari
	\thanks{This work is supported by the UK Research and Innovation council's \emph{Future Leaders Fellowship} program ``Sensing Beyond Barriers'' (MRC Fellowship award no.~MR/S034897/1).}
	\thanks{D.~Florescu and A.~Bhandari are with the Dept. of Electrical and Electronic Engineering, Imperial College London, South Kensington, London SW7 2AZ, UK. (Emails: \texttt{d.florescu@imperial.ac.uk} and \texttt{ayush@alum.mit.edu})}

\thanks{F.~Krahmer is with the Dept. of Mathematics, TU Munich, Boltzmannstra{\ss}e 3, 85748 Garching, Germany. (Email: \texttt{felix.krahmer@tum.de})}
}

\markboth{Manuscript}
{Shell \MakeLowercase{\textit{et al.}}: Bare Demo of IEEEtran.cls for IEEE Journals}

\maketitle

\setstretch{1.35}
\begin{abstract}
The Unlimited Sensing Framework (USF) was recently introduced to overcome the sensor saturation bottleneck in conventional digital acquisition systems.
At its core, the USF allows for high-dynamic-range (HDR) signal reconstruction by converting a continuous-time signal into folded, low-dynamic-range (LDR), modulo samples.
HDR reconstruction is then carried out by algorithmic unfolding of the folded samples.
In hardware, however, implementing an ideal modulo folding requires careful calibration, analog design and high precision. 
At the interface of theory and practice, this paper explores a computational sampling strategy that relaxes strict hardware requirements by compensating them via a novel, mathematically guaranteed recovery method.
Our starting point is a generalized model for USF. 
The generalization relies on two new parameters modeling \emph{hysteresis} and \emph{folding transients} in addition to the modulo threshold.
\emph{Hysteresis} accounts for the mismatch between the reset threshold and the amplitude displacement at the folding time and we refer to a continuous transition period in the implementation of a reset as \emph{folding transient}.
Both these effects are motivated by our hardware experiments and also occur in previous, domain-specific applications. 
We show that the effect of hysteresis is beneficial for the USF and we leverage it to derive the first recovery guarantees in the context of our generalized USF model. Additionally, we show how the proposed recovery can be directly generalized for the case of lower sampling rates. Our theoretical work is corroborated by hardware experiments that are based on a hysteresis enabled, modulo ADC testbed comprising off-the-shelf electronic components.
Thus, by capitalizing on a collaboration between hardware and algorithms, our paper enables an end-to-end pipeline for HDR sampling allowing more flexible hardware implementations. 
\end{abstract}

\begin{IEEEkeywords}
Analog-to-digital conversion (ADC), modulo sampling, HDR sensing, Shannon sampling theory, thresholding.
\end{IEEEkeywords}

\IEEEpeerreviewmaketitle

\newpage
\setstretch{1}
\tableofcontents

\newpage
\setstretch{1.35}

\section{Introduction}

The well-known Shannon sampling theory, governing the current analog-to-digital converters (ADCs), shows how a signal of any bandwidth can be acquired by sampling it uniformly at a sufficiently high rate. If the sampling criterion is not satisfied, the input is distorted by a phenomenon known as aliasing. At the same time, ADCs have a limited dynamic range, beyond which the circuit exhibits distortions leading to saturated samples. {Since the ADCs are at front end of data acquisition systems, signal saturation or clipping poses a fundamental bottleneck in the sampling pipeline. To overcome this barrier,} recently, the Unlimited Sensing Framework (USF) was conceptualized \cite{Bhandari:2017,Bhandari:2020:P,Bhandari:2020a}. This alternative procedure based on a modulo architecture allows for  recovery of signals with amplitudes beyond the dynamic range. A modulo-ADC based hardware was recently reported in \cite{Bhandari:2021:J} and experiments based on this hardware indeed validated the practical capability of the USF approach. 

\subsection{Overview of the Unlimited Sensing Framework}

The USF enables the co-design of hardware and algorithms for high-dynamic-range (HDR) signal recovery.
This is done via an encoding which is consistent with `folding', known as \emph{modulo folding} that resets the output $z\rb{t}$ in the analog domain every time it crosses one of the thresholds $\pm\lambda$, thus ensuring that $z\rb{t}\in\sqb{-\lambda,\lambda}$. Subsequently the output is sampled uniformly with period $T$.

To this end, the Unlimited Sampling Theorem and its variants for different signal classes \cite{Bhandari:2018,Bhandari:2018a,Bhandari:2020:C}, were the first recovery guarantees akin to Shannon's sampling theorem adapted to modulo samples. The key takeaway is that different signal classes can be recovered from modulo samples without requiring any side-information such as residuals, or knowledge of a finite set of samples unaffected by the modulo non-linearity. Among the previous approaches an attempt is made in \cite{Rieger:2008:J} for modest folding using a variant of Itoh's method.
In particular, this is in stark contrast to self-reset ADC based CMOS imagers \cite{Sasagawa:2015:J}, where the authors proceed by separating the input into folded data and the residual signal (typically referred to as the reset count), but require that both of them are being recorded in order to allow for reconstruction. 

A central result in USF is that an input signal $g\rb{t}$ of arbitrary amplitude, bandlimited to $\Omega\ \mathrm{rad/s}$, can be recovered from modulo measurements $z\rb{kT}$ if $T<{1}/{(2\Omega e)}$. {In the noisy scenario the recovery is still possible but the amplitude of $g\rb{t}$ is restricted \cite{Bhandari:2020a}.} The original recovery approach, which we will refer to as Unlimited Sampling Algorithm ($\mathsf{USAlg}$), convolves the data with a finite difference filter, and assumes that the modulo folding duration is $0$. In this way, the input can be written as $g\rb{t}=z\rb{t}+\varepsilon_g\rb{t}$, where $\varepsilon_g\rb{t}$ is a residual function spanning a discrete set of values $\varepsilon_g\rb{t}\in 2\lambda \mathbb{Z}$. Given that the finite differences of the filter lie on a uniform grid, $\mathsf{USAlg}$ annihilates the effect of $\varepsilon_g\rb{t}$ via an additional modulo operation. 
For a hardware validation of $\mathsf{USAlg}$, we refer to \cite{Bhandari:2021:J}.

\subsection{The Need for a Generalized Model for Modulo Sampling}
\label{subsect:motivation}
 
In this paper, we advocate a generalized model for a modulo sampling architecture. The generalization arises by introducing two new degrees-of-freedom in addition to the modulo threshold $\rb{\lambda}$, namely, \emph{hysteresis} $\rb{h}$ and \emph{folding transient} $\rb{\alpha}$. While the mathematical model will be made explicit in the sequel, here we clarify the role of these model parameters and the underlying motivation. 

Electronic hardware is typically designed to a technical specification. The more stringent the design constraints, the more expensive is the resulting hardware. This is because the electronic hardware is the front-end of all acquisition systems and {any loss of information} during data acquisition would compromise the algorithmic capability. Our work is guided by the \emph{computational sensing} philosophy where the core idea is to compensate imprecise hardware by sophisticated algorithms that can account for the model mismatch in cost-effective hardware implementations. In the context of our work on unlimited sampling, this is  achieved by allowing for non-negligible \emph{hysteresis} and \emph{folding transients} and incorporating them in a refined model.
Our approach is motivated by phenomenological observations both in our hardware implementation and in certain related works. In the following we explain those concepts in detail, and then discuss the phenomenological observations. 

\begin{figure}[!t]
	\centerline{\includegraphics[width=.55\textwidth]{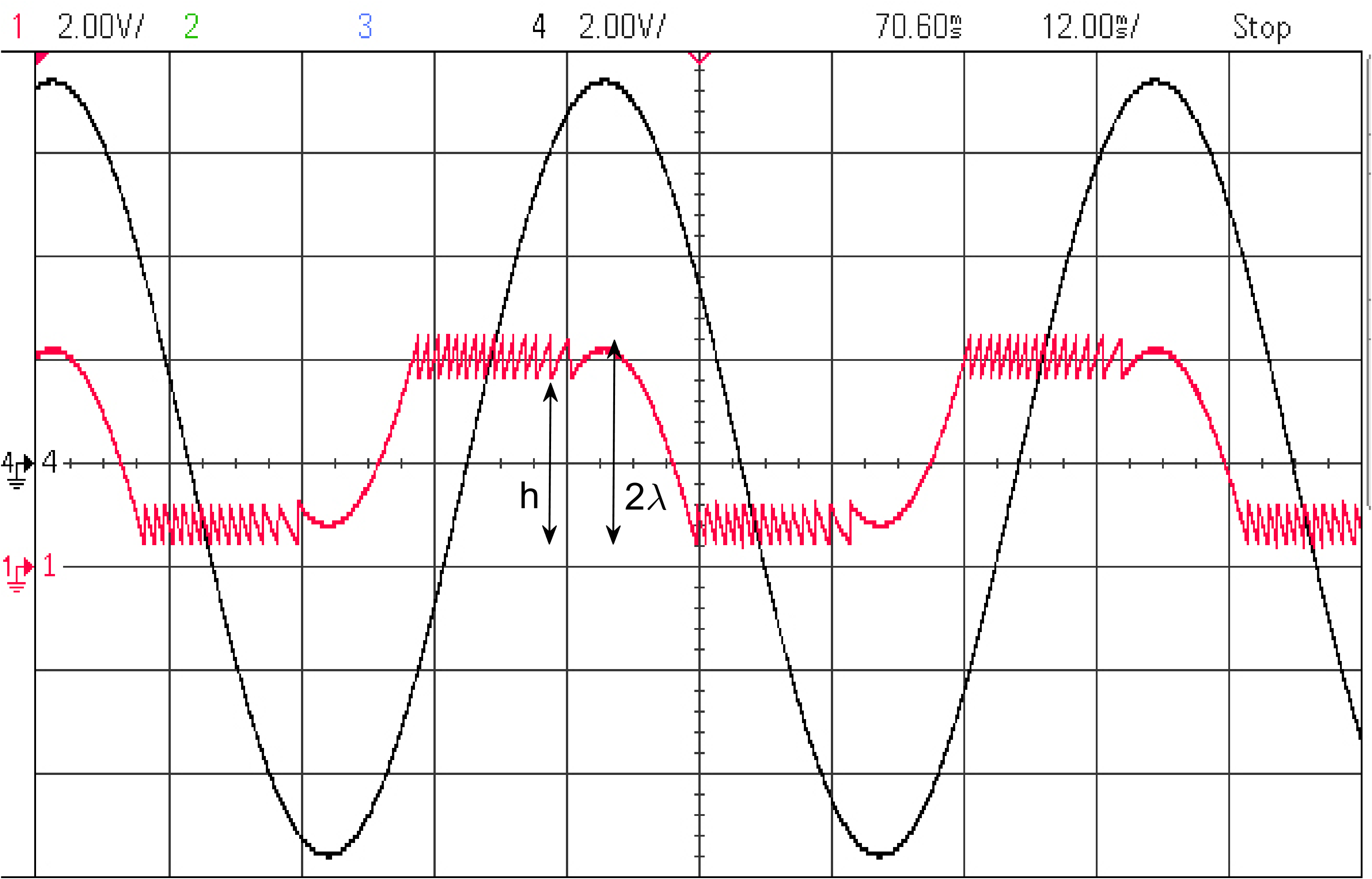}}
	\caption{Oscilloscope screenshot of modulo samples (red) generated by our hardware prototype with threshold $\lambda$ showing the effect of hysteresis quantified by parameter $h$, which we will exploit in this paper.
	A live {YouTube} demonstration with a different values of $h$ is available at \href{https://youtu.be/R4rji5jOjD8}{\texttt{https://youtu.be/R4rji5jOjD8}}.}
	\label{fig:measured_data}
\end{figure}
\begin{figure}[t!]
	\if\figmode\PH
	\centerline{\includegraphics[width=.5\columnwidth]{PH_c.jpg}}
	\else
	\centerline{\includegraphics[width=0.5\columnwidth]{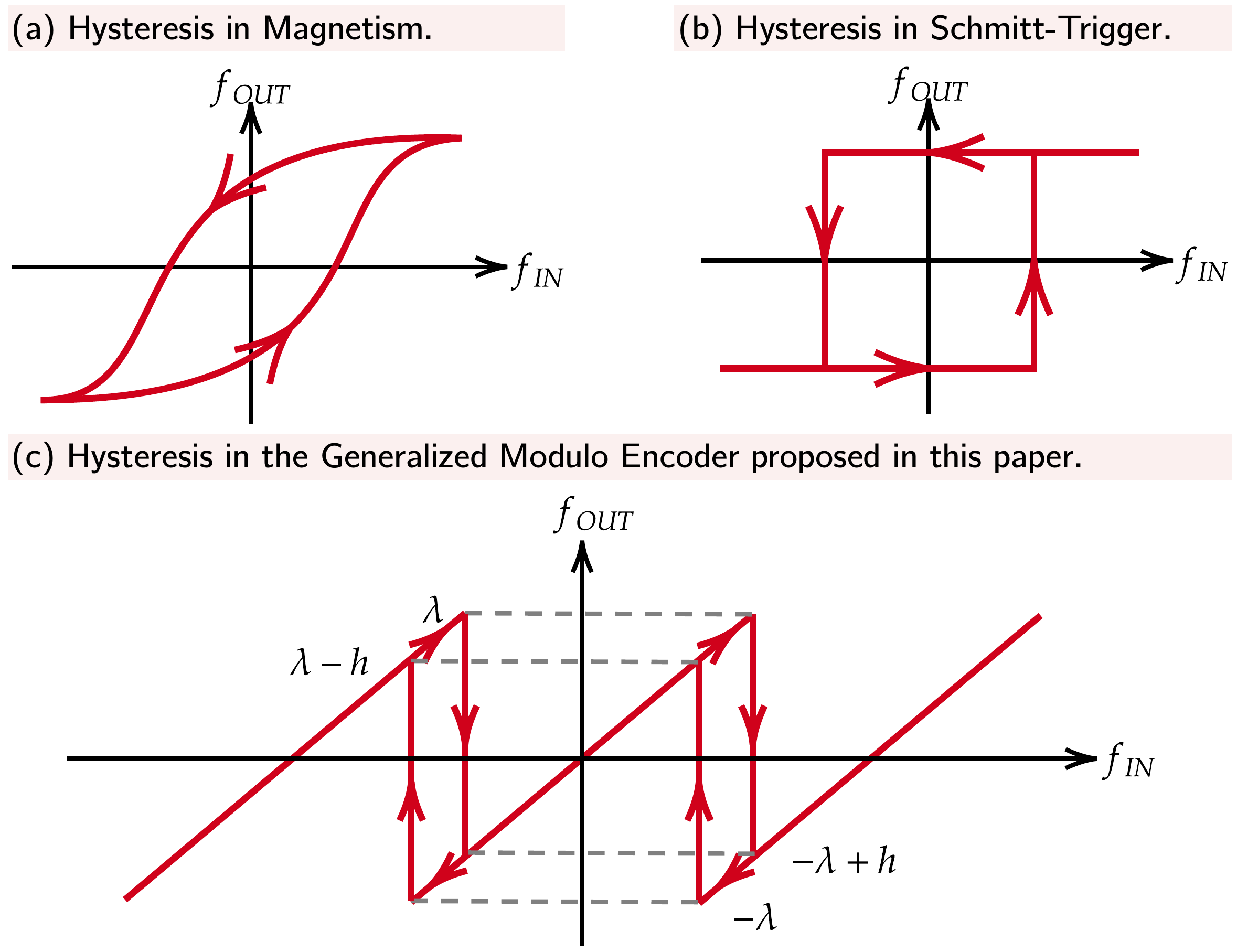}}
	\fi
	\caption{ {Plots of the hysteresis transfer functions observed in a range of applications. (a) The hysteresis curve in magnetism. (b) A Schmitt trigger exhibiting hysteresis. (c) Hysteresis in the context of our \emph{generalized model} for modulo sampling.}}
	\label{fig:hysteresis}
\end{figure}

\medskip

\begin{enumerate}[label = $\bullet$]
\itemsep 10pt
\item {\bf Hysteresis.} For approaches in \cite{Bhandari:2020a,Bhandari:2017} to work, one needs to ensure that the reset is hardwired so that folding kicks-in at $\pm \lambda$ and shifts the signal to the opposite threshold $\mp \lambda$. For this to be accurately implemented in hardware, one requires precise calibration and prudent choice of components. 
If the reset threshold and post-reset value are not perfectly aligned, the measurements will no longer follow an exact modulo model. As we will argue, the model mismatch endows the encoder with memory, where the output at a given time depends on the history of the input. We will call this memory effect \emph{hysteresis}\footnote{The name is inspired from the duality with the hysteresis phenomenon in ADC conversion, see below for more details.} and quantify it by a parameter $h$ (see \fig{fig:measured_data} and \fig{fig:hysteresis}).  The algorithms presented in this paper adapt to this effect, and hence allow to relax the calibration step 
giving rise to more flexible hardware design options.

\item {\bf Folding Transients.}
When implementing modulo non-linearity using electronic components such as comparators, physical non-idealities due to \emph{slew-rate} play a role; the implication is that exact transitions can not be realized, resulting in a non-negligible transition slope. Minimizing the slope duration requires strict design specifications. The implementation cost can be reduced by allowing for a longer folding transient. To put this idea in perspective, we draw an analogy to the notion of \emph{transition bandwidth} of a filter or the maximal \emph{slew rate} of an electronic component. The smaller the transition bandwidth, or the higher the slew rate, the closer one can get to the ideal model, but the costlier the implementation. Very much analogously, we observed during the evolution of our hardware and experiments, that there is a trade-off between the demanding circuit design necessary to approximate ideal folds and the model mismatch resulting from folding transients in a simpler circuit design. This motivated us to investigate recovery approaches where transients are not just considered as artefacts, but rather incorporated into the model for high performance even for low-complexity implementations.

\item {\bf Hysteresis and Folding Transients in other Works.} In a broader context, specific instances of our generalised model match the observations in a number of other works in different fields. In particular, hardware implementations for both ECG recordings \cite{Rieger:2008:J} and neural recordings \cite{Chen:2014} fold the amplitude back to zero when reaching a certain threshold, which directly corresponds to a hysteresis parameter of $h=\lambda$ in our model. At the same time, both implementations also give rise to folding transients. These are, however, not exploited in a refined model, but treated as artefacts (for example, \cite{Chen:2014} takes a heuristic route of identifying them and removing the corresponding samples from the data set). Hysteresis is a well-known non-ideality in ADC design \cite{Pereira:1999:C,IEEEstd:2017}, particularly acting on the quantization process, which was then exploited in circuits such as the Schmitt trigger, where it is voluntarily added for stabilization purposes. Given that the modulo ADC output can be interpreted as quantization noise \cite{Bhandari:2020:C}, it is affected by a hysteresis phenomenon that is inverse to the hysteresis observed in the quantized signal. That is why hysteresis and folding transients do not represent a significant problem in the case of the self-reset ADCs, because in this case both residuals and the "folded" signal are acquired \cite{Bhandari:2020a} and the deviations from the ideal model, present in both the residual and folded signal, get canceled out when added together. For example the work in \cite{Park:2010:C} uses hysteresis to stabilize the folding mechanism, but it assumes that the residual is known. Generally, by storing the full dynamic range of the residual, self-reset ADCs do not employ full dynamic range recovery.

\end{enumerate}

\subsection{Contributions}
 {In this paper, we formulate and investigate a generalized model for modulo sampling as motivated in the previous subsection. Beyond theoretical developments, we use our customized hardware as a testbed for experimental validation, thus demonstrating that our findings are key for a practical realization of a flexible, \emph{end-to-end} modulo sensing pipeline. More concretely, we
\begin{enumerate}[label = $\bullet$]
    \item precisely formulate an acquisition model accounting for hysteresis and folding transients.
    \item design a reconstruction algorithm adapted to this generalized sampling model.
    \item rigorously derive recovery guarantees for this algorithm under the refined model, including robustness guarantees.
    \item validate the algorithmic performance numerically.
    \item experimentally demonstrate that our method accurately recovers signals for a variety of hysteresis parameters via a modulo hardware testbed based on off-the-shelf components that can implement variable hysteresis.
\end{enumerate}}

\subsection{Related Work}

The conceptualization of the USF triggered a number of follow-up studies. A wider class of measurement scenarios were covered including computed tomography \cite{Bhandari:2020:Ca,Beckmann:2021:J}, sensor array processing \cite{FernandezMenduina:2020:C,FernandezMenduina:2021:J}, and event-driven sampling \cite{Florescu:2021:C}. For bandlimited signals, which are the main focus in our work, the folded measurements are proven to be injective: a bandlimited function is mapped one-to-one to its modulo samples for sampling rates just above the critical Nyquist rate \cite{Bhandari:2019a, Romanov:2019:J}. The USF methodology was extended to specific classes of inputs, such as sparse signals or sinusoidal mixtures \cite{Bhandari:2018,Bhandari:2018a}.  {Reconstruction in the context of compressed sensing from  modulo measurements was studied in\cite{Musa:2018}; identifiability results were presented recently in \cite{Prasanna:2020}. Sparse recovery from modulo samples limited to two modulo folds was discussed in \cite{Shah:2021:J}.}
Further on, Unlimited Sampling was extended to one-bit uniform samples \cite{Graf:2019}. The USF reconstruction was demonstrated theoretically and numerically for samples corrupted by bounded noise \cite{Bhandari:2020a}, and also for imaging data \cite{Bhandari:2020:C}. A different approach performs denoising directly on modulo data prior to recovery \cite{Cucuringu:2018,Cucuringu:2020:J}. The phase unwrapping problem can recover modulo folded data, but it represents a simplified case of USF recovery \cite{Itoh:1982:J,Ghiglia:1998:J,Choi:2007:J}. {A Fourier domain approach for recovering modulo data with arbitrarily close folding times was recently introduced in \cite{Bhandari:2021:J}. The new approach allows for recovery when the modulo threshold is unknown. However, the associated recovery guarantees assume input periodicity and that the number of folds is known \emph{a priori}. Therefore the work in \cite{Bhandari:2021:J} is complementary to the one in current manuscript. In \fig{fig:hystfig}, we show the positioning of the proposed method among related works.}

\begin{figure}[!t]
\centerline{\includegraphics[width=0.65\columnwidth]{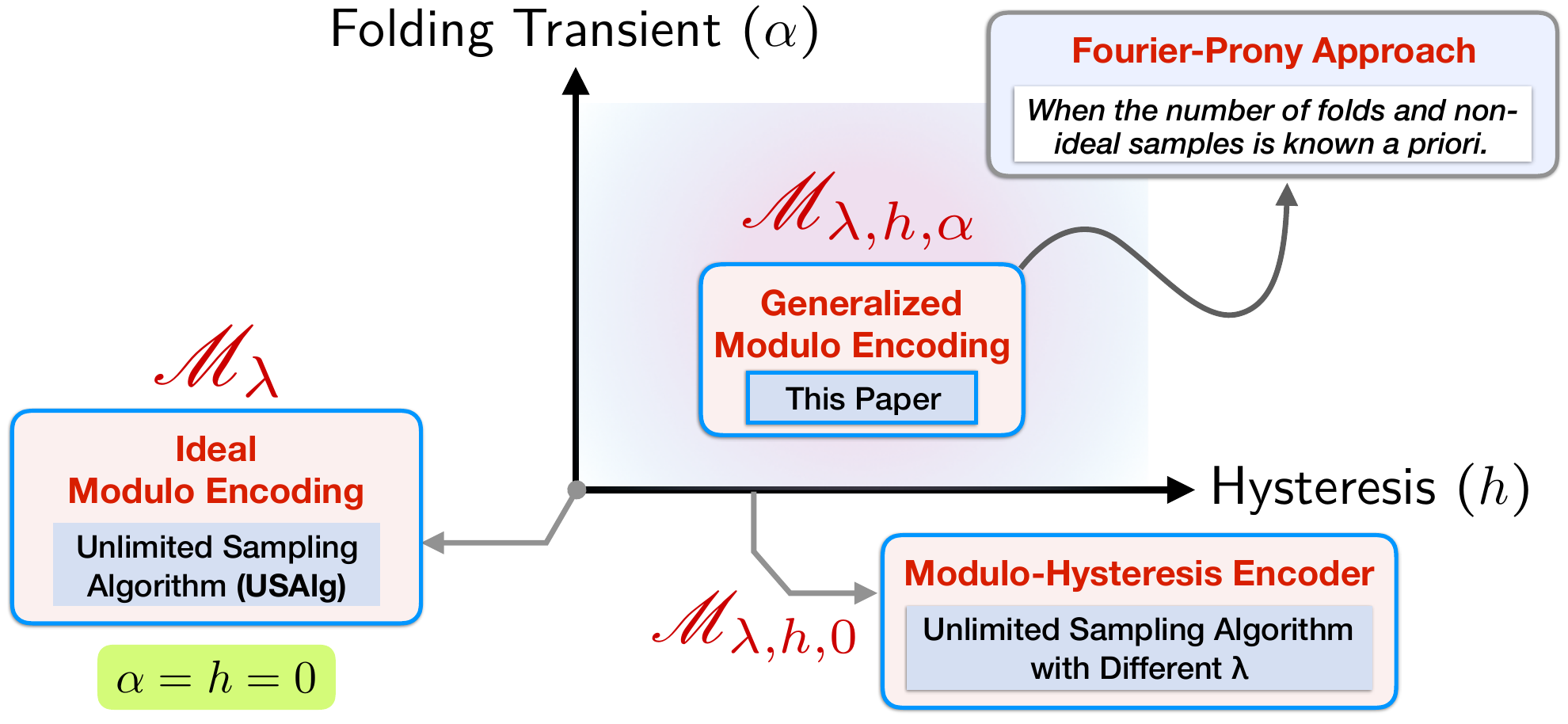}}
\caption{The positioning of the proposed model and recovery approach among related methods.}
\label{fig:hystfig}
\end{figure}

Thresholding is a well-known approach for solving inverse problems \cite{Donoho:1995,Tibshirani:1996, Daubechies:2004:J}. Its practical utility in USF was shown by filtering the data with a kernel such as a B-spline \cite{Graf:2019} or a wavelet \cite{Rudresh:2018}. However, given that the folding times of the modulo encoder can get arbitrarily close for any input, no work so far has been able to provide guarantees in this context (see Section \ref{sect:USF}).

\subsection{Notation}
We use $\mathbb{Z}$ and $\mathbb{R}$ to denote the set of integers and real numbers, respectively, and $\mathbb{N}$ to denote the set of positive integers. Continuous functions are denoted as $f\rb{t}$ and discrete sequences as $f\sqb{k}$. The Hilbert space of square-integrable functions is denoted as $L^2\rb{\mathbb{R}}$, and the corresponding sequence space in discrete time is $\ell^2$. The norm in any space $H$ is denoted as $\norm{f}_H$. For invertible continuous functions $f\rb{t}$, the inverse is denoted as $f^{-1}$ s.t. $f^{-1}\rb{f\rb{t}}=t$. The derivative of order $\const{N}$ is denoted as $f^{\rb{\const{N}}}\rb{t}$ and, for sequences, the finite difference of order $\const{N}$ is $\Delta^\const{N} f \sqb{k}$, which is computed by applying recursively $\Delta^{N+1} f \sqb{k}= \Delta^N f\sqb{k+1}- \Delta^N f\sqb{k}$, where $\Delta^{0} f = f$.
A right-inverse operation of $\Delta^1$ is the cumulative summation 
$\mathsf{S}$ defined as $\mathsf{S}f\sqb{k}=\sum\nolimits_{i=-\infty}^k f\sqb{i}$. In this paper, we will only apply this operation to finitely supported sequences, so summability is not an issue.
The space of square integrable functions bandlimited to $\Omega$ is the Paley-Wiener space denoted $\PW{\Omega}$. The indicator function $\ind_{S}\sqb{k}$ is $1$ for $k\in S$ and $0$ otherwise. The floor and ceiling functions are $\floor{\cdot}$ and $\ceil{\cdot}$, respectively, and $[\![ x ]\!]=x- \lfloor x \rfloor$ denotes the fractional part of $x$. The set $\sqb{a_1,a_2}$ consists of all real numbers between $a_1$ and $a_2$ including $a_1, a_2$. Similarly, $\rb{a_1,a_2}$ is the set $\sqb{a_1,a_2}$ excluding values $a_1$ and $a_2$. Furthermore, let $ \mathrm{supp} \rb{f}$ denote the support of sequence $f\sqb{k}$, let $\vb{S}$ denote the cardinality of a set $S$, and let $\emptyset$ denote the empty set. Whenever possible we use the same notation for a ground truth quantity and the corresponding reconstruction, with a tilde $\sim$ added to the latter, e.g., $\widetilde{f}$ denotes the reconstruction corresponding to a function or sequence $f$. By $f_\infty$ we denote the absolute sequence norm $\norm{f}_{\ell^\infty}$ or the absolute function norm $\norm{f}_{L^\infty}$, respectively. Let $\textsf{MSE}\rb{f_1,f_2}$ denote the mean square error between finite sequences $f_1\sqb{k},f_2\sqb{k}$ of length $\const{N}_f$ defined as, 
$$\textsf{MSE}\rb{f_1,f_2}=\frac{1}{\const{N}_f} \sum\limits_{k=1}^{\const{N}_f} \rb{f_1\sqb{k}-f_2\sqb{k}}^2.$$

\subsection{Organization of the Paper}
 {
The proposed modulo encoder and associated reconstruction framework are given in Section \ref{sect:New_Modulo}. Section \ref{sect:numerical_demonstration} presents numerical results with synthetic data as well as experimental data based on our modulo sampling hardware prototype. The proofs to our main results are in Section \ref{sect:Proofs}. We conclude with a discussion in Section \ref{sect:conclusion}.
}

\section{Modulo with Hysteresis and Transients}
\label{sect:New_Modulo}
Here we start by discussing the ideal modulo encoder and its shortcomings in Section~\ref{sect:USF}.
Section~\ref{sect:prop_mod_desc} then precisely formulates  our new model that incorporates two additional degrees of freedom to quantify hysteresis and folding transients, motivated in Section~\ref{subsect:motivation}.
We discuss the limitations of using $\mathsf{USAlg}$ for this model in Section \ref{sect:hyst_USFrec}. In Section \ref{sect:thresholding}, we present our new recovery method with finite difference filters.

\subsection{The Ideal Modulo Encoder}
\label{sect:USF}
For an input function $g\rb{t}$, the ideal modulo encoder with threshold $\lambda$ is denoted by $\MO$, and defined as
\cite{Bhandari:2017}, \cite{Bhandari:2018}
\begin{equation}
	\MO \rb{g\rb{t}} = 2\lambda \left( {\fe{ {\frac{g\rb{t}}{{2\lambda }} + \frac{1}{2} } } - \frac{1}{2} } \right).
	\label{eq:fold}
\end{equation}
The input is typically assumed to be bandlimited $g \in \PW{\Omega} $. In the noiseless scenario the USF guarantees that the input of the ideal modulo encoder can be recovered from the output samples provided that the sampling period satisfies $T<\frac{1}{2\Omega e}$. Furthermore, recovery guarantees have also been provided in the case of data corrupted by bounded noise \cite{Bhandari:2020a}. The USF approach relies on the observation that $\varepsilon_g \rb{t}\triangleq g\rb{t}-\MO \rb{g\rb{t}}\in 2\lambda\cdot \mathbb{Z}$, where $\varepsilon_g \rb{t}$ is known as the residual function. In other words, for the ideal modulo encoder the values of $\varepsilon_g \rb{t}$ lie on an equally spaced grid with step $2\lambda$. However, this is not true for non-ideal modulo encoders exhibiting phenomena such as folding transients, leading to reconstruction distortions.

 {
A different reconstruction approach for the ideal modulo encoder consists of filtering the modulo output with a high-pass kernel and \emph{thresholding} the filtered signal to detect the folding times. The filtering operation cancels out the bandlimited input and results in high amplitude pulses centered in the folding times \cite{Graf:2019}. Given that the pulses have a non-zero support size, this approach introduces the challenge of ensuring a minimal separation between folding times.} However, for the ideal modulo encoder the input bandwidth $\Omega$ cannot be used directly to infer a minimal separation. For example, $g\rb{t}$ may have a local minimum right below the threshold $\lambda$, which can render two arbitrarily close folding times (\fig{fig:modulo_reset_times1}).

\begin{figure}[!t]
\centerline{\includegraphics[width=.6\textwidth]{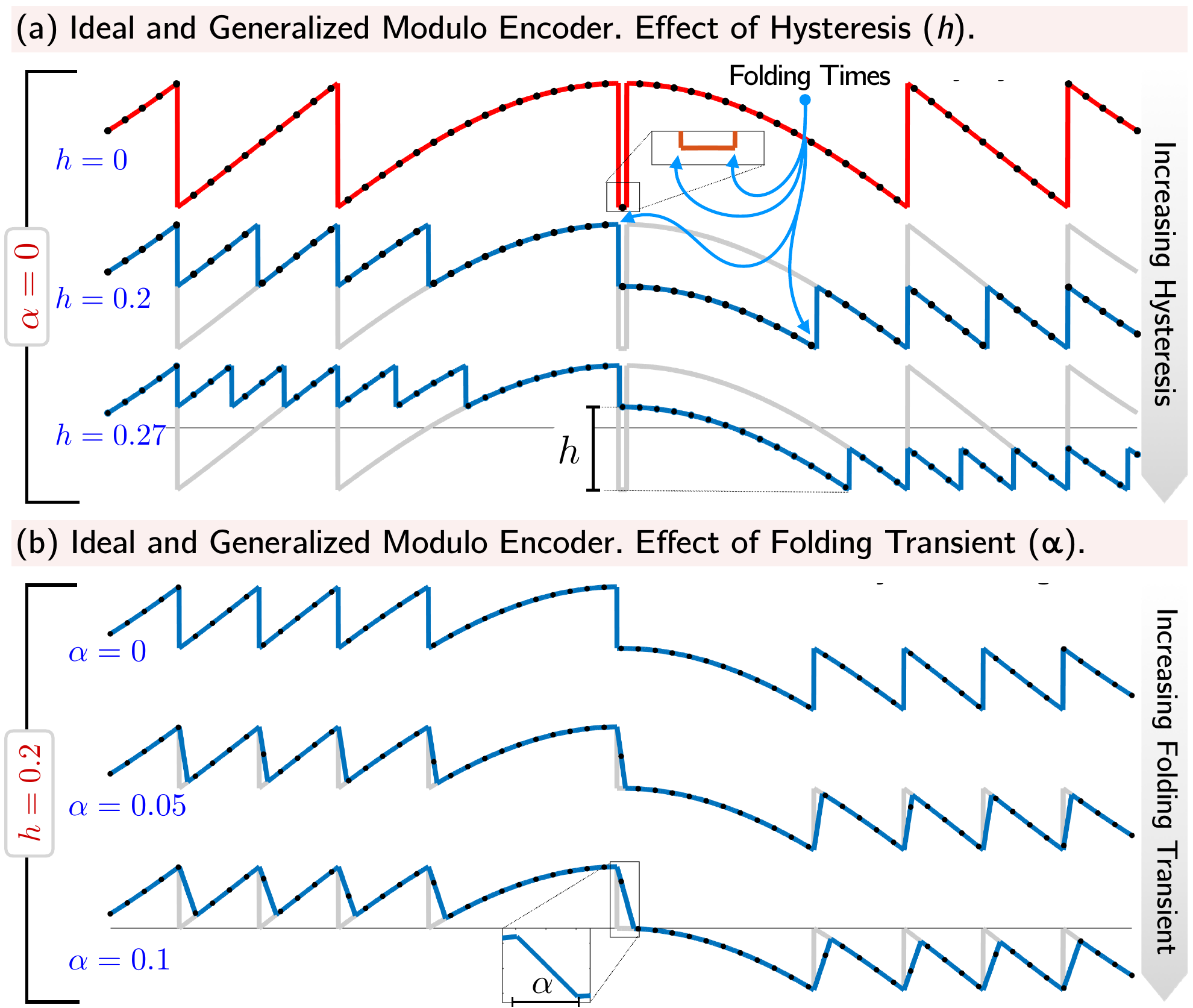}}
\caption{ {Encoding a $\mathrm{sinc}$ function with the ideal and generalized modulo encoder for different values of the hysteresis parameter $h$ (up) and of the transient duration $\alpha$ (down). Observe that the ideal modulo encoder can lead to very close folding times, while the generalized modulo encoder induces a significantly larger separation between folding times.}}
\label{fig:modulo_reset_times1}
\end{figure}

The work in \cite{Rudresh:2018} uses Lasso to detect the folding times, which is closely linked to thresholding, but does not introduce theoretical guarantees. In \cite{Graf:2019} the authors denoted the shortest interval in which a bandlimited function can cross a threshold twice as \emph{superoscillation parameter}. However, this is not possible to compute in the general case where only the bandwidth and maximum amplitude are known. None of the thresholding based reconstructions for the ideal modulo encoder were tested on examples where consecutive folding times are very close, as in \fig{fig:modulo_reset_times1}.

\subsection{The Proposed Generalized Modulo Encoder}
\label{sect:prop_mod_desc}
We now formally define the generalized modulo encoder.
\begin{definition}[Generalized Modulo Encoder]The analog modulo encoder with threshold $\lambda$, hysteresis parameter $h\in\left[0,2\lambda\right)$ and transient parameter $\alpha$, is an operator $\MOh:\PW{\Omega} \rightarrow L^2\rb{\mathbb{R}}$, $\boldsymbol{\mathsf{H}}=\sqb{\lambda,h,\alpha}$, that generates an asynchronous sequence $\cb{\tau_p}$ and an analog function $z \rb{t}=\MOh g \rb{t},$ $ t\geq \tau_0$ in response to input $g$. The sequence $\cb{\tau_p}$ is computed recursively as
	\begin{align*}
		\tau_1&=\min\cb{t > \tau_0 \vert \MO\rb{g\rb{t}+\lambda}= 0}\\
		\tau_{p+1}&=\min\cb{t> \tau_p \vert \MO\rb{g\rb{t}-g\rb{\tau_p}+h s_p}=0},
	\end{align*}
	where $s_p= \mathrm{sign} \rb{g\rb{\tau_p}-g\rb{\tau_{p-1}}}, p\geq1$.
	The output function $z \rb{t}$ is then given by 
	\begin{equation}
		\label{eq:zt}
		z \rb{t}=g\rb{t}-\varepsilon_g\rb{t},
	\end{equation}
	where 
$$\varepsilon_g\rb{t}=\sum\limits_{p\in\mathbb{Z}} s_p \varepsilon_0\rb{t-\tau_p},\ t\in\mathbb{R}$$ with $$\varepsilon_0\rb{t}=2\lambda_h \sqb{ \rb{1/\alpha} t\cdot \ind_{\left[0,\alpha\right)}\rb{t}+ \ind_{\left[\alpha,\infty\right)}\rb{t}} \quad \mbox{and} \quad \lambda_h=\lambda-h/2.$$
	\label{def:cont_modulo}
\end{definition}

The output of the proposed encoder is depicted in \fig{fig:modulo_reset_times1} for a bandlimited input. The hysteresis parameter $h$, defining the post-reset values $\cb{\lambda-h, -\lambda+h}$, ensures a minimal amplitude change is required to trigger a new reset.  As we will establish in Lemma \ref{lem:separation} below, this determines a minimal separation between the folding times that is not satisfied by the ideal modulo encoder (\fig{fig:modulo_reset_times1}). To ensure a sufficient separation, we thus suggest a component/parameter choice that avoids the unlikely scenario that $h$ is very small. Given the wide range of feasible parameter choices, this will be a rather simple task, in  contrast to the fine calibration to predetermined values required in typical self-reset or modulo ADCs.

In Definition \ref{def:cont_modulo}, $\varepsilon_0\rb{t}$ models the folding transient as a line with gradient $1/\alpha$. This ensures that $\varepsilon_g\rb{t}$ is mathematically continuous, unlike for the ideal modulo, which is in accordance with the maximal slew rate requirement. Furthermore, the hysteresis parameter $h$ controls the distance between each pair of modulo thresholds $\cb{\lambda, \lambda-h}, \cb{-\lambda, -\lambda+h}$. We note that $\MOh$ is not a  {static nonlinear system} like is the case for the ideal modulo $\MO$. Therefore the value of $\MOh g\rb{t}$ does not depend only on $g\rb{t}$, but on its whole history $g\rb{s}, s\leq t$.  {Moreover, for the purposes of avoiding clipping and saturation, the input $g\rb{t}$ to the encoder is an analog continuous signal rather than a discrete sequence of samples.}

The output of $\MOh g\rb{t}$ is subsequently passed through an ADC with sampling period $T$  {whose dynamic range is $\sqb{-\lambda,\lambda}$.} We denote the input samples by $\gamma\sqb{k}=g\rb{kT}$ and the output samples by $y\sqb{k}=\MOh g\rb{kT} $. Furthermore, we denote $\varepsilon_\gamma\sqb{k}=\varepsilon_g\rb{kT}$. We assume that $T\geq \alpha$, which guarantees a maximum of one sample located on the transient. This is satisfied in experimental settings, even for sampling periods much smaller than the USF requirement $T\ll\frac{1}{2\Omega e}$.
Here we consider that the output samples are corrupted by bounded noise, i.e., $y_\eta\sqb{k}=y\sqb{k}+\eta\sqb{k}$, where $\eta\sqb{k}$ is the noise sequence satisfying $|\eta\sqb{k}|\leq \eta_\infty$. It is easy to show that, if $h, \alpha \rightarrow 0$ we have $y_\eta\sqb{k}=\MO\rb{\gamma\sqb{k}}+\eta\sqb{k}$ \cite{Bhandari:2017,Bhandari:2018}, which proves that the proposed model is backward compatible with the ideal modulo encoder. 

To quantify the separation between the folds, we define the \textit{discrete folding times} $n_p$ as the integer values satisfying ${n_p} = \left\lceil {{\tau _p}/T} \right\rceil $. 
We then obtain the following estimate for the same.

\begin{lem}[Minimum Separation between Folds]
\label{lem:separation}
Let $y_\eta\sqb{k}=\MOh g\rb{kT}+\eta\sqb{k}, h\in[0,2\lambda[$, $g\in\PW{\Omega}$ and $\{n_p\}_{p\in\mathbb{N}}$, such that ${n_p} = \left\lceil {{\tau _p}/T} \right\rceil $. Then the following holds
\begin{equation*}
n_{p+1}-n_p\geq \floor{\frac{h^\ast}{T\Omega g_\infty}}, \quad h^\ast=\min\cb{h,2\lambda-h}.
\end{equation*}
\end{lem}
\begin{proof}
We calculate $g\rb{\tau_{p+1}}$ using the first order Taylor expansion around anchor point $\tau_p$
\begin{equation*}
g\rb{\tau_{p+1}}=g\rb{\tau_{p}}+\rb{\tau_{p+1}-\tau_p}g^{\rb{1}}\rb{\tau_p^*}, \tau_p^*\in\sqb{\tau_p,\tau_{p+1}}. 
\end{equation*}
From Definition \ref{def:cont_modulo} and Bernstein's inequality ${g_\infty^{\left( n \right)} } \leqslant {\Omega}{g_\infty }$,
\begin{equation}
\label{eq:min_dist_analog_folds}
\tau_{p+1}-\tau_p=\vb{\frac{g\rb{\tau_{p+1}}-g\rb{\tau_{p}}}{ g^{\rb{1}}\rb{\tau_p^*}}}\geq\frac{h^\ast}{\Omega g_\infty}.
\end{equation}
Furthermore, a direct calculation yields, 
$\left\lceil {{\tau _{p + 1}}/T} \right\rceil - \left\lceil {{\tau _p}/T} \right\rceil  \geqslant \left\lfloor {\left( {{\tau _{p + 1}}/T} \right) - \left( {{\tau _p}/T} \right)} \right\rfloor  \geqslant \left\lfloor {{h^*}/\left( {T\Omega {g_\infty }} \right)} \right\rfloor$.
\end{proof}

\subsection{Reconstruction via Unlimited Sampling}
\label{sect:hyst_USFrec}
 {When there are no or only negligible folding transients, one can recover $\widetilde{\gamma}\sqb{k}$, which is equal to  ${\gamma}\sqb{k}$ up to an integer multiple of $2\lambda$, from $y_\eta\sqb{k}=\MOh g\rb{kT}+\eta\sqb{k}, h\in[0,2\lambda)$ using a variation of $\mathsf{USAlg}$. }
This can be seen via the following argument: If we sample \eqref{eq:zt} with period $T$ it follows that
\begin{equation}
\label{eq:modulo_samples}
y_\eta\sqb{k}=\gamma_\eta\sqb{k}+\sum\limits_{p\in\Z} s_p \varepsilon_0\rb{kT-\tau_p}.
\end{equation}
Given that $n_p T-\tau_p \in \left[0,T\right)$ and thus $\varepsilon_0\rb{kT-\tau_p}$ may take any value in interval $\sqb{0,2\lambda_h}$, $\mathsf{USAlg}$ is not directly applicable. However, if $\alpha\rightarrow0$, then $\varepsilon_0\rb{kT-\tau_p}\rightarrow2\lambda_h \ind_{\left[\tau_p,\infty\right)}\rb{kT} \in \cb{0,2\lambda_h}$. Along the same lines as for $\mathsf{USAlg}$, the recovery condition is $\vb{\Delta^\const{N} \gamma\sqb{k} + \Delta^\const{N} \eta\sqb{k}}<\lambda_h, \forall k \in \mathbb{Z}$, which is guaranteed by 
\begin{equation}
\label{eq:gamma_bound_hyst}
\rb{T\Omega e}^\const{N} g_\infty + 2^\const{N} \eta_\infty<\lambda_h.
\end{equation}
Then $\widetilde{\gamma}\sqb{k}$ can be computed using
\begin{equation}
\label{eq:USF_modulo_rec_hyst}
\Delta^\const{N} \gamma_\eta \sqb{k}=\mathscr{M}_{\lambda_h} \rb{\Delta^\const{N} y_\eta\sqb{k}}.
\end{equation}
 {This approach exploits that for negligible folding transient it is unlikely that a sample lies directly on the transient.  For that reason, for non-negligible values of  $\alpha$, $\mathsf{USAlg}$ will no longer produce reasonable results.}
More precisely, if $\alpha>0$, then computing $\mathscr{M}_{\lambda^\ast} \rb{\Delta^\const{N} \MOh g\rb{kT}+\eta\sqb{k}}$ does not annihilate the residual $\varepsilon_\gamma\sqb{k}$ for any $\lambda^\ast$, as opposed to the ideal modulo in Section \ref{sect:USF}. In this case, the samples during the transient period introduce significant distortions, which are further amplified when attempting to integrate iteratively $\Delta^\const{N} \gamma_\eta$ as in $\mathsf{USAlg}$  \cite{Bhandari:2017,Bhandari:2020a}, rendering the reconstruction unstable. 

\subsection{Reconstruction via Thresholding}
\label{sect:thresholding}
Here we employ a thresholding approach to recovery. There have been a number of isolated attempts recovering the input of an ideal modulo encoder based on thresholding. However, as explained in Section \ref{sect:USF}, a major drawback that does not allow introducing recovery guarantees is that the folding times of the ideal USF encoder can be arbitrarily close. Given samples $y_\eta\sqb{k}=\gamma_\eta\sqb{k}+\varepsilon_\gamma\sqb{k}$, the main challenge is to recover both $\gamma_\eta\sqb{k}$ and $ \varepsilon_\gamma\sqb{k}$. For an ideal modulo encoder, the USF approach annihilates the residual $\varepsilon_\gamma\sqb{k}$. Here we convolve the data with a filter $\psi_N\sqb{k}\triangleq\Delta^N\sqb{k+1}$, and then identify the residual by thresholding the samples $\psi_N\ast y_\eta\sqb{k}$, which requires identifying the folding times $\tau_p$ and signs $s_p$ \eqref{eq:modulo_samples}.
We first introduce sufficient conditions under which the detection of the folding times is possible via thresholding. Subsequently, we provide theoretical guarantees for computing $\tau_p$ and $s_p$.
The filtered samples $\psi_N \ast y_\eta$ satisfy
\begin{equation}
	\label{eq:psi_y}
	\psi_N \ast y_\eta\sqb{k}=\psi_N\ast\gamma_\eta\sqb{k} - \psi_N\ast\varepsilon_\gamma\sqb{k}.
\end{equation}
In order to identify the modulo folds via thresholding, $\psi_N$ is designed to enhance the effect of $\varepsilon_\gamma$ in \eqref{eq:psi_y}. The following lemma (proof, Section~\ref{sect:Proofs}) shows that
the non-zero values of $\psi_N\ast\varepsilon_\gamma\sqb{k}$ are centered around the discrete folding times $n_p$.
\begin{lem}[Support of Filtered Residual]
\label{lem:supp_psi_eps}
Let $\psi_N\sqb{k}=\Delta^{N}\sqb{k+1}$ and let 
$\Sz=\bigcup\nolimits_{p\in\Z} {\Sz^p}$, where $\Sz^p=\cb{n_p - \rb{N-1}, \dots, n_p+1}$. Then
\begin{equation}
\label{eq:supp_psi_eps}
\mathrm{supp}\rb{\psi_N\ast\varepsilon_\gamma}\subseteq \mathbb{S}_N.
\end{equation}
\end{lem}

\subsubsection{The Case of a Single Fold}
For simplicity, we first assume a single folding time $\tau_1>0$, such that $\Sz=\Sz^1=\cb{n_1-\rb{N-1},\dots,n_1 +1}$. 
The key idea is to exploit smoothness to obtain that 
\begin{equation}
\label{eq:condIa}
\norm{\psi_N\ast \gamma_\eta}_{\ell^\infty}< \frac{\lambda_h}{2N},
\end{equation}
and to observe that far away from the jump, this condition will also hold for $y_\eta$ in place of $\gamma_\eta$. That is, the set 
\begin{equation}
\mathbb{M}_N\triangleq\cb{k\in\Z\ \vert \ \vb{\psi_N\ast y_\eta\sqb{k}} \geq \frac{\lambda_h}{2N}}\label{eq:MN}
\end{equation}
is localized around the jump and can be used to identify the jump location up to a minimal ambiguity. The following lemma (proved in Section \ref{sect:Proofs}) makes this precise. 
\begin{lem}
\label{lem:k1_k2}
Assume \eqref{eq:condIa} is satisfied, and let 
$\subSz=\cb{k_m, k_M}$,
\begin{equation}\label{eq:k1_k2}
k_m=\min \mathbb{M}_N, \ \  k_M=\max \mathbb{M}_N.
\end{equation}
Then $\subSz\subseteq\Sz$, $k_m\in\cb{n_1-N+1,n_1-N+2}$, $k_M\in\cb{n_1 ,n_1 +1}$, and for $\mathbb{I}_1=\sqb{\tau_1+ \frac{\alpha}{2N},\tau_1+\alpha \rb{1-\frac{1}{2N}}}$,
\begin{align}
n_1 T&\in \mathbb{I}_1 \Rightarrow k_M-k_m=\const{N},	\label{eq:mid_transient}\\
n_1 T&\in \left[\tau_1,\tau_1+\alpha\right) \setminus  \mathbb{I}_1  \Rightarrow k_M-k_m\in\cb{\const{N}-1,\const{N}}\label{eq:start_end_transient},\\ 
n_1 T&\geq \tau_1+\alpha \Rightarrow k_M-k_m=N-1.\label{eq:after_transient}
\end{align}
\end{lem}

The following theorem (proved in Section \ref{sect:Proofs}) resolves the remaining ambiguity and also gives a means to identify the direction of the fold, i.e., the sign $s_1$ in Definition~\ref{def:cont_modulo}.

\begin{theo}[Estimation of Folding Times]
\label{theo:s1_tau1}
Assume that $\norm{\psi_N\ast \gamma_\eta}_{\ell^\infty}<{\lambda_h}/\rb{2N}$ and let $k_m,\ k_M\in \Z$ defined as in \eqref{eq:k1_k2}. Furthermore, let $\widetilde{n}_1$ and $\widetilde{s}_1$ be the estimations of the discrete folding time and sign defined as
\begin{align}
\label{eq:n_1}
\widetilde{n}_1&=k_M-1,\\
\label{eq:s_1}
\widetilde{s}_1&=-\mathrm{sign} \rb{\psi_N\ast y_\eta\sqb{k_m}}.
\end{align}
\begin{enumerate}[leftmargin = *, label = $\alph*)$]
\item  If $k_M-k_m=\const{N}$, consider the estimated folding time 
\end{enumerate}
\begin{equation}
\label{eq:folding_seq}
\widetilde{\tau}_1\triangleq\widetilde{n}_1 T-\alpha\cdot\frac{\psi_N\ast y_\eta\sqb{\widetilde{n}_1}+2\lambda_h \widetilde{s}_1 \rb{N-1}}{2\lambda_h \widetilde{s}_1 N}.
\end{equation}
Then the following hold
\begin{equation}
\label{eq:tau_p_bound}
\widetilde{n}_1=n_1, \quad \widetilde{s}_1=s_1, \quad\vb{\widetilde{\tau}_1-\tau_1} < \frac{\alpha}{4N^2}.
\end{equation}
\begin{enumerate}[leftmargin = *, label = $\alph*)$,start=2]
\item  If $k_M-k_m=\const{N}-1$, consider the estimated folding time $ \widetilde{\tau}_1=\widetilde{n}_1 T$. Then $\widetilde{s}_1=s_1$, and we have two cases
\begin{enumerate}[leftmargin = *, label = $b_\arabic*)$,start=1]
\item $n_1T \in \left[ \tau_1,\tau_1+ \frac{\alpha}{2N}\right)$
\begin{equation}
\label{eq:bound_start_transient}
\widetilde{n}_1=n_1, \quad
\vb{\widetilde{\tau}_1-\tau_1}<\frac{\alpha}{2N}.
\end{equation}
\item $ n_1T > \tau_1+\alpha \rb{1-\frac{1}{2N}}$
\begin{equation}
\label{eq:bound_end_transient}
\widetilde{n}_1=n_1-1, \quad \vb{\widetilde{\tau}_1-\tau_1}<T-\alpha+\frac{\alpha}{2N}.
\end{equation}
\end{enumerate}
\end{enumerate}
\end{theo}

Theorem \ref{theo:s1_tau1} shows the benefit of the transients for accurate reconstruction of the folding instants. When a sample in the central part of the transient is available, the associated folding location can be identified up to an accuracy of $\tfrac{\alpha}{4 N^2}$, while only an accuracy of order $\tfrac{\alpha}{2 N}$ or even $T-\alpha\tfrac{2N-1}{2N}$ can be achieved when all the sampling locations lie near the edges or outside of the transient. The result always recovers $\tau_1$ with an error strictly smaller than $T$, which is the smallest error that can be guaranteed with $\mathsf{USAlg}$. Accurate folding time estimates then also yield accurate signal reconstruction via the formula

\begin{equation}
\widetilde{\gamma}\sqb{k}=y_\eta\sqb{k}+s_1 \varepsilon_0\rb{kT-\widetilde{\tau}_1}.
\label{eq:input_rec}
\end{equation}

\subsubsection{The Case of Multiple Folds}
Generally there are $P$ folding times $\cb{\tau_p}_{p=1}^P$, where $P$ is unknown. Here, computing $\widetilde{\gamma}\sqb{k}$ can be performed by computing $\widetilde{\tau}_p$ and $\widetilde{s}_p$ with Theorem \ref{theo:s1_tau1} sequentially, if the supports of each two filtered folds are disjoint, guaranteed by $\Sz^{p_1}\cap\Sz^{p_2}=\emptyset, \forall p_1,p_2=[1,P], p_1\neq p_2$. In other words, this requires the separation of the folding times $n_{p+1}-n_p \geq \const{N}+1$, which is guaranteed via Lemma \ref{lem:separation} if
\begin{equation}
\frac{h^\ast}{T\Omega g_\infty} \geq \const{N}+1.
\label{eq:separation} 
\end{equation}
Therefore conditions \eqref{eq:condIa} and \eqref{eq:separation} are sufficient to guarantee the recovery of $s_p,\tau_p, p\in [1,P]$ sequentially, using Theorem \ref{theo:s1_tau1} as follows. Let $\subSz=\bigcup\nolimits_{p \in \left[ {1,P} \right]} {\cb{k_m^p, k_M^p}}$, where the pairs $\cb{k_m^p,k_M^p}$ are computed sequentially as follows. For $p=1$,
\begin{align}
\label{eq:k_m1_k_M1}
\begin{split}
k_m^1&=\min\cb{k\ \vert \ \vb{\psi_N\ast y_\eta\sqb{k}} \geq \frac{\lambda_h}{2N}},\\
k_M^1&=\max\cb{k\leq k_m^1+\const{N} \ \vert \ \vb{\psi_N\ast y_\eta\sqb{k}} \geq \frac{\lambda_h}{2N}}.
\end{split}
\end{align}
For $p=2,\dots,P$
\begin{align}
\label{eq:k_mp_k_Mp}
\begin{split}
k_m^p&=\min\cb{k> k_M^{p-1} \ \vert \ \vb{\psi_N\ast y_\eta\sqb{k}} \geq \frac{\lambda_h}{2N}},\\
k_M^p&=\max\cb{k\leq k_m^p+\const{N} \ \vert \ \vb{\psi_N\ast y_\eta\sqb{k}} \geq \frac{\lambda_h}{2N}}.
\end{split}
\end{align}
Then $\widetilde{s}_p, \widetilde{\tau}_p, p=1,\dots,P$, computed sequentially with Theorem \ref{theo:s1_tau1} from $\subSz$ satisfy
\begin{equation}
\widetilde{s}_p=s_p, \vb{\widetilde{\tau}_p-\tau_p}<\max\cb{\frac{\alpha}{2N}, T-\alpha\frac{2N-1}{2N}}.
\label{eq:final_bound_taup}
\end{equation}
The following lemma provides sufficient recovery conditions for ${s}_p, {\tau}_p, p=1,\dots,P$.
\begin{lem}[Sufficient recovery conditions] 
\label{lem:sufficient_cond}
\begin{enumerate}[leftmargin=0.5cm,label=${\alph*}$)] 
\item The values $\cb{s_p,\tau_p}_{p=1}^P$ can be recovered from $y_\eta\sqb{k}$ with an error bounded by \eqref{eq:final_bound_taup} if
\begin{enumerate}[leftmargin=1.5cm,label=$\mathrm{TH}_{\arabic*}$)]
\item $\rb{T\Omega e}^\const{N} g_\infty + 2^\const{N}\eta_\infty\leq\frac{\lambda_h}{2N}$,
\item $\rb{\const{N}+1} T\Omega g_\infty \leq h^\ast$,\quad $h^\ast=\min\cb{h,2\lambda-h}$
\end{enumerate}
\item Furthermore, $\mathrm{TH}_1)$ and $\mathrm{TH}_2)$ are satisfied if
\begin{enumerate}[leftmargin=35pt,label=$b_{\arabic*}$)]
\item In the noiseless scenario ($\eta_\infty=0$),
\begin{equation}
\label{eq:Nbound1}
N \leq \frac{h^*}{4 e T \Omega g_\infty}-1,
\end{equation}
\item In the noisy scenario ($\eta_\infty\neq 0$), we require $\eta_\infty\leq \frac{\lambda_h}{8}$,
\begin{gather}
\label{eq:Nbound21}
N \leq \min\cb{\frac{h^*}{8 e T \Omega g_\infty}-1,f^{-1}\rb{\frac{\lambda_h}{\eta_\infty}}},
\end{gather}
where $f:\mathbb{R}_+\rightarrow \mathbb{R}_+, f(x)=x\cdot 2^{x+2}$. 
\end{enumerate}
\end{enumerate}
\end{lem}
\begin{proof}
\begin{enumerate}[leftmargin=0.5cm,label=$\alph*$)]
\item $\left.{\mathrm{TH}}_1\right)$ is a sufficient condition for \eqref{eq:condIa} via \cite{Bhandari:2020a}. $\left.{\mathrm{TH}}_2\right)$ follows from \eqref{eq:separation}.
\begin{enumerate}[leftmargin=35pt,label=$b_{\arabic*}$)]
\item Here, for $\forall N\in\Z$, $\left.{\mathrm{TH}}_1\right)$ and $\left.{\mathrm{TH}}_2\right)$ are equivalent to
\begin{gather}
T\leq\min\cb{ B_1\rb{N}, B_2\cb{N}},\label{eq:Tbound1}\\
B_1\rb{N}=\frac{1}{\Omega e \sqrt[N]{2N\cdot g_\infty/\lambda_h}}, \ \ B_2\rb{N}= \frac{h^*}{\rb{N+1}\Omega g_\infty},\label{eq:Tbound2}
\end{gather}
It can be shown that $B_1\rb{N}\geq B_2\rb{N}  {\lambda_h}/\rb{2h^* e}, \forall N\geq 1$, and thus $T \leq  {C}/\rb{\rb{N+1}\Omega g_\infty}$, where $C=\min\cb{h^*,\frac{\lambda_h}{2e}}$. The result follows due to $C
\geq {h^*}/{4e}$.

\item The choice of $N$ ensures that $2^N\eta_\infty\leq \frac{\lambda_h}{4N}$. Furthermore,  $\rb{T\Omega e}^N g_\infty \leq\frac{\lambda_h}{4N}$, and thus $\left.{\mathrm{TH}}_1\right)$ and $\left.{\mathrm{TH}}_2\right)$ hold if
\end{enumerate}
\end{enumerate}
\[
T\leq\min\cb{ B_1^\eta\rb{N}, B_2\cb{N}}, \ 
B_1^\eta\rb{N}=\frac{{{{\left( {4N{g_\infty/\lambda_h }} \right)}^{ - \frac{1}{N}}}}}{{\Omega e}}.
\]
Here, it can be shown that $B_1^\eta\rb{N}\geq B_2\rb{N}\cdot \frac{\lambda_h}{4h^* e}$ and thus $T \leq C_\eta/\rb{N+1} \Omega g_\infty, C_\eta=\min\cb{h^*,\frac{\lambda_h}{4e}}$, and the proof follows due to $C_\eta\geq {h^*}/{8e}$.
\end{proof}

In the noiseless scenario $\eta_\infty=0$, $\left.{\mathrm{TH}}_1\right)$ is always satisfied for $\const{N}$ large enough given that the left-hand side decreases exponentially. We note that the USF condition is more relaxed \eqref{eq:gamma_bound_hyst}. However, USF is intrinsically incompatible with the case where $\alpha\neq 0$, which is always true in practice. 
Condition $\left.{\mathrm{TH}}_2\right)$ shows how $h>0$ ensures there are at least $N$ samples in between each two folds. In the following, we show how the choice of $h$ impacts the distancing of the folds in practice. We generate an input $g\rb{t}=\mathrm{sinc}\rb{t}-0.8499$ which allows arbitrarily close folds as in \fig{fig:modulo_reset_times1}. The modulo threshold is $\lambda=0.15$. The output is depicted in \fig{fig:small_h} for $h=0$ (gray) and for $h=5\cdot 10^{-3}$ (blue). Condition $\left.{\mathrm{TH}}_1\right)$ is satisfied for $N=2$, but $h=0$ does not satisfy the separability condition, which requires a minimum of $h=5\cdot 10^{-3}\approx 3\% $ of $\lambda$.
\begin{figure}[!t]
    \centering
    \includegraphics[width=.45\textwidth]{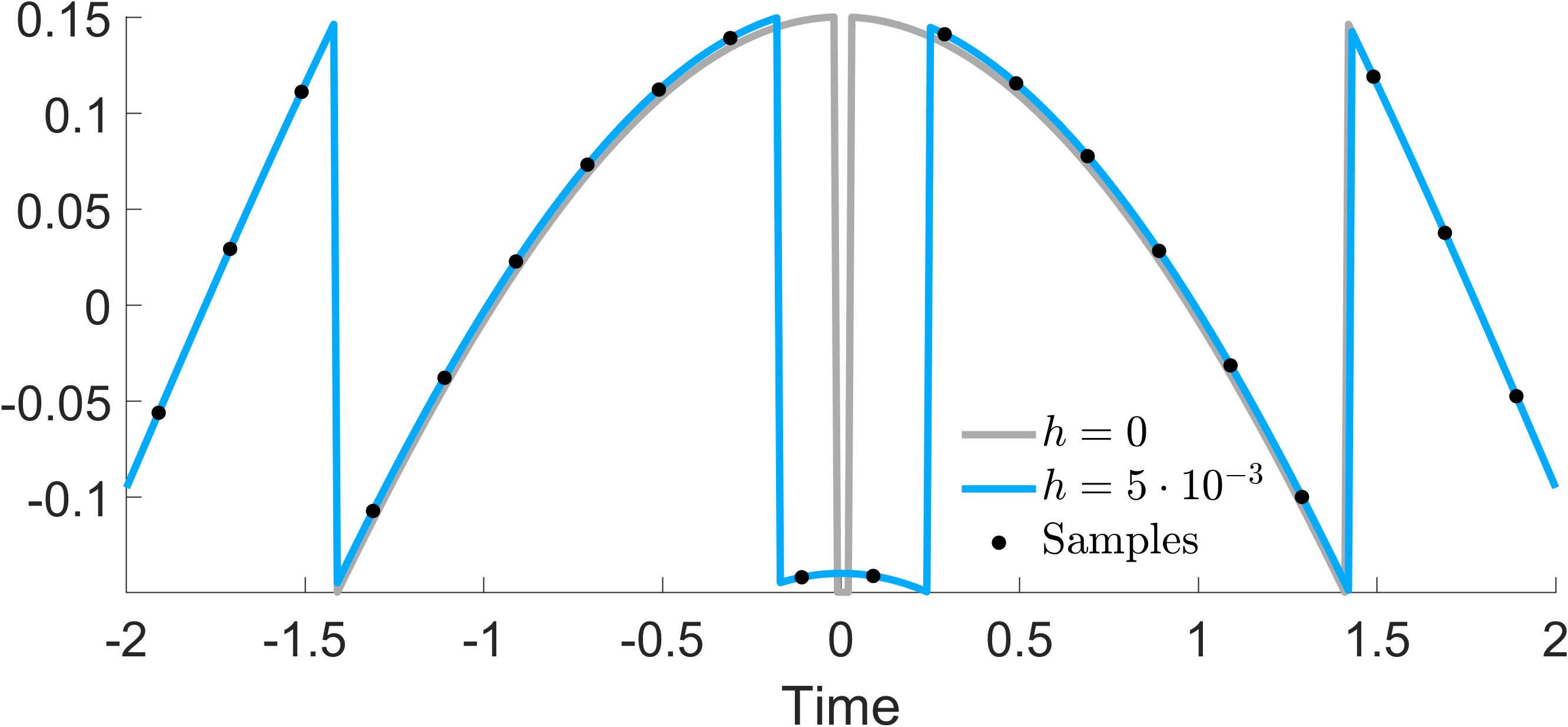}
    \caption{The minimal $h$ allowing at least two samples between each folds for a $\mathrm{sinc}\rb{t}$ input function.}
    \label{fig:small_h}
\end{figure}

\subsubsection{Input Reconstruction}

The values of $\widetilde{s}_p, \widetilde{\tau}_p$ are used to estimate the input samples $\widetilde{\gamma}$ as outlined in Algorithm \ref{alg:1}.
\begin{algorithm}[!t]
\label{alg:1}
\SetAlgoLined
{\bf Data:} $y_\eta\sqb{k},\psi_N\sqb{k},\lambda,h,\alpha,\Omega>0$.\\
\KwResult{ $\widetilde{\gamma}\sqb{k}$. }
 \begin{enumerate}[leftmargin =*, label = $\arabic*)$]
\item Compute the filtered samples $\psi_N\ast y_\eta\sqb{k}$.
\item Compute the numebr of folds $P$ and values $\cb{k_m^p,k_M^p}_{p=1}^P$ with \eqref{eq:k_m1_k_M1}, \eqref{eq:k_mp_k_Mp}.
\item For $p=1,\dots,P$
\begin{enumerate}[leftmargin =*, label = $4.\arabic*)$]
\item If $k_M^p-k_m^p=\const{N}$:\\
Compute $\widetilde{n}_p, \widetilde{s}_p$ with \eqref{eq:n_1},\eqref{eq:s_1} and the folding time estimation $\widetilde{\tau}_p=\widetilde{\tau}^{n_p}_p$ with \eqref{eq:folding_seq}.
\item If $k_M^p-k_m^p=\const{N}-1$:\\ 
Compute $\widetilde{n}_p, \widetilde{s}_p$ with \eqref{eq:n_1}, \eqref{eq:s_1} and the folding time estimation $\widetilde{\tau}_p=\widetilde{n}_p T$.
\end{enumerate}
\item Compute $\widetilde{\varepsilon}_\gamma\sqb{k}=\sum\nolimits _{p=1}^P s_p \varepsilon_0\rb{kT-\widetilde{\tau}_p}$,\\ 
where
$\varepsilon_0\rb{t}=2\lambda_h \sqb{ \rb{1/\alpha} t\cdot \ind_{[0,\alpha[}\rb{t}+ \ind_{[\alpha,\infty[}\rb{t}}$.
\item Compute $\widetilde{\gamma} \sqb{k} = y_\eta \sqb{k}+\widetilde{\varepsilon}_\gamma\sqb{k}$.
\end{enumerate}
\caption{Recovery Algorithm.}
\end{algorithm}
The following result (proved in Section \ref{sect:Proofs}) bounds the input reconstruction error.
\begin{prop}[Error Bound for Input Recovery]
\label{prop:mse}
Assuming that the conditions in Lemma \ref{lem:sufficient_cond} a) or b) are satisfied, the constant offset $2n \lambda$ via \eqref{eq:input_rec} is known, $T\geq \alpha+\rb{{\alpha}/{4N^2}}$, and the data consists of $K$ samples and $P$ folding times, the input reconstruction $\widetilde{\gamma}$ resulted from Algorithm \ref{alg:1} satisfies 
\begin{equation}
\mathsf{MSE}\rb{\widetilde{\gamma},\gamma}\leq \frac{\lambda_h^2 }{N^2}\cdot \frac{P}{K},
\end{equation}
where $n\in\Z$ and $P$ is the number of folds. If $P$ is unknown, by assuming $h^*\leq KT\Omega g_\infty$,
\begin{equation}
\label{eq:mse2}
\mathsf{MSE}\rb{\widetilde{\gamma}, \gamma}\leq \frac{\lambda_h^2 }{N^2}\cdot \frac{2 T \Omega g_\infty}{h^*}.
\end{equation}
\end{prop}

{It follows that the bound can be decreased by selecting the largest $N$ admitted by Lemma \ref{lem:sufficient_cond}.} In practice, however, $N$ can be further increased as long as $\left.{\mathrm{TH}}_1\right)$ is satisfied, given that the separation between the folding times, guaranteed by $\left.{\mathrm{TH}}_2\right)$, can be directly inferred from the data, as shown in  Section \ref{sect:numerical_demonstration}. The following theorem derives new bounds on the reconstruction error.

\begin{theo}[Noiseless Case]
Assume that the conditions in Lemma \ref{lem:sufficient_cond} b) are true, where $N$ is the largest integer satisfying \eqref{eq:Nbound1}, the offset $2n\lambda$ is known, $8eT\Omega g_\infty \leq h^*\leq KT\Omega g_\infty$ and $T\geq 2\alpha$. Assuming $\eta_\infty=0$, it follows that
\begin{equation}
\mathsf{MSE}\rb{\widetilde{\gamma},\gamma}\leq C\cdot T^3,
\end{equation}
where $C=2 \rb{16e\lambda_h}^2\cdot\rb{{ \Omega g_\infty}/{h^*}}^3$
\begin{proof}
We substitute \eqref{eq:Nbound1}, which is satisfied for $N = \floor{\frac{h^*}{4 e T \Omega g_\infty}-1}$, in \eqref{eq:mse2}. Given that $\tfrac{h^*}{4 e T \Omega g_\infty}-1\geq \tfrac{h^*}{8 e T \Omega g_\infty}$, we then use that $\floor{x}\geq \frac{x}{2}, \forall x\in\left[1,\infty\right)$.
\end{proof}
\end{theo}

\begin{theo}[Noisy Case]
Assume that the conditions in Lemma \ref{lem:sufficient_cond} b) are true, where $N$ is the largest integer satisfying \eqref{eq:Nbound21}. Furthermore, assume the constant offset $2n\lambda$ is known, $16eT\Omega g_\infty \leq h^*\leq KT\Omega g_\infty$ and $T\geq 2\alpha$. If the data is corrupted by noise sequence $\eta\sqb{k}$ satisfying $\vb{\eta\sqb{k}}\leq\eta_\infty$, where $\eta_\infty\leq\frac{\lambda_h}{8}$, then the error satisfies
\begin{gather}
\mathsf{MSE}\rb{\widetilde{\gamma},\gamma}\leq \max\cb{ C\cdot T^3, C_\eta \cdot T},
\end{gather}
where $$C=2 \rb{32 e\lambda_h}^2\cdot\rb{\frac{ \Omega g_\infty}{h^*}}^3, \quad  C_\eta=\frac{8\lambda_h^2}{\rb{f^{-1}\rb{\frac{\lambda_h}{\eta_\infty}}}^2} \cdot\frac{\Omega g_\infty}{h^*}.$$
\end{theo}
\begin{proof}
Follows from \eqref{eq:Nbound21} and \eqref{eq:mse2}.
\end{proof}

\section{Sampling at Low Rates}

The most conservative of the two conditions in Lemma \ref{lem:sufficient_cond} assumes there is a contiguous region of $N$ unfolded samples in between each two folds. Here we show that in fact we need just a single such region, which entails that the proposed algorithm can be extended to much lower sampling rates than the ones in Lemma \ref{lem:sufficient_cond}.
A similar assumption is also made in \cite{Ordentlich:2018} for the case of ideal modulo. For now we assume that the unfolded samples are the first $N$; later we generalize to arbitrary locations. 
Outside of this region, we only require that there is  one sample between each two folds $\cb{\tau_p,\tau_{p+1}}$, as will be explained next. 

The overarching idea is that, starting from the filtered data $y_{1,N}\sqb{l}\triangleq \psi_N \ast y_\eta\sqb{l}, l\in\cb{1,\dots,K}$, one can pass through the samples from left to right and compute recursively the sequence $y_{k,N}\sqb{l}, k=2,\dots,K$, where $y_{k,N}\sqb{l}$ contains no contribution of the past folding times $\tau_1,\dots,\tau_p$ satisfying $\tau_p<kT$. There on, given that there is a sample between each two folds, it follows that $y_{k,N}\sqb{k}$ may contain information about at most one folding time, which can be exploited to compute $\tau_p, s_p$. To define sequence $y_{k,N}\sqb{l}$, we note that $\psi_N \ast y_\eta\sqb{k}=\psi_N \ast \gamma\sqb{k} - \sum_{p=1}^P e_{p,N}\sqb{k}$, where, as shown in the proof of Lemma \ref{lem:supp_psi_eps},
	\begin{equation}
		\label{eq:erN}
		\begin{aligned}
			e_{p,N}\sqb{k}\triangleq 2\lambda_h s_p \left({\beta}_p d_{n_p}^{N-1}\sqb{k}\right.
			\left.+\rb{1-{\beta}_p}d_{n_p+1}^{N-1}\sqb{k}\right)
		\end{aligned}
	\end{equation}
where $d_{n_p}^{N-1}\sqb{k}\triangleq \Delta^{N-1}\sqb{k-{n}_p} $. We note that in general, for ${\beta}_p\neq0,1$, the function in \eqref{eq:erN} is an affine linear combination of two kernels of the form $2\lambda_h s_p d_{n_p}^{N-1}\sqb{k}$ centered in neighboring samples. We aim to rewrite $\sum_{p=1}^P e_{p,N}\sqb{k}$ as a sum of shifted kernels, where each kernel is centered in a distinct sample. To ensure that this is possible, we assume $T<\min\cb{\tau_{p+1}-\tau_p}-\alpha$, which guarantees there are no two adjacent transient samples. Then
\begin{equation}
\label{eq:erN2}
    \sum\limits_{p=1}^P e_{p,N}\sqb{k}= 2\lambda_h\sum\limits_{m=1}^K b_m s_{p_m} d_{m}^{N-1}\sqb{k},
\end{equation}
where $b_m\geq 0$ and $p_m$ are defined as follows. If $m=n_p$ for some $p$ then $p_m=p$. Furthermore, $\beta_p=1\Rightarrow b_m=1$, and  $\beta_p\in\rb{0,1} \Rightarrow b_m=\beta_p, b_{m+1}=1-\beta_p$. If $m\neq n_p\ \forall p$ then $b_m=0, p_m=1$. We then have a one-to-one mapping between the unknown pairs $\cb{b_m, s_{p_m}}_{m=1}^K$ and samples $m\in\cb{1,\dots K}$. Similarly to Theorem \ref{theo:s1_tau1}, $b_m$ and $s_m$ can be identified via thresholding sequentially, provided that $\norm{\psi_N\ast \gamma_\eta}_{\ell^\infty}<\tfrac{\lambda_h}{2N}$. 
We define
\begin{equation}
    y_{k,N}\sqb{l}=\psi_N \ast y_\eta\sqb{l} + 2\lambda_h \sum\limits_{m=1}^{k-1} b_m s_{p_m} d_{m}^{N-1} \sqb{l}.
\end{equation}
The values of $b_k, s_{p_k}$ are then computed from $\bar{y}_k\triangleq y_{k,N}\sqb{k}$ via thresholding. Specifically, if $\vb{\bar{y}_k}> 2\lambda_h-\theta$, $\theta=\frac{\lambda_h}{2N}$, then sample $k$ corresponds to a complete fold with $\beta_p=0,1$, and if $\theta<\vb{\bar{y}_k}< 2\lambda_h-\theta$ then $k$ corresponds to the transient of fold $p$, such that $b_k\in\cb{\beta_p,1-\beta_p}$. To distinguish between the two cases, we note that, if $b_k=\beta_p, b_{k+1}=1-\beta_p$ then $\vb{\bar{y}_{k}+\bar{y}_{k+1}-2\lambda_h \mathrm{sign} \rb{\bar{y}_k}}$ is small, which is quantified via tolerance $\theta_\beta$. Otherwise, we conclude that $b_{k-1}=\beta_p, b_{k}=1-\beta_p$. If $b_k=\beta_p$, then we can estimate $\widetilde{\beta}_p=\frac{\vb{\bar{y}_k}}{2\lambda_h}$. The sign $s_{p_k}$ is computed as in Theorem \ref{theo:s1_tau1} as $s_{p_k}=-\mathrm{sign} \rb{\bar{y}_k}$. The values $b_k, s_{p_k}$ are subsequently used to compute sequence $y_{k+1,N}\sqb{l}$, and the process continues recursively (see Algorithm \ref{alg:2}).

In the general case that the region of $N$ unfolded samples does not lie in the beginning, but say in between $\cb{\tau_p,\tau_{p+1}}$, the idea is to follow the approach sketched above both to the left and to the right, starting from an anchor point sample $k_0$ such that $y\sqb{k}=g\rb{kT}+M, k\in\cb{k_0,\dots,k_0+N}$ for some $M\in\mathbb{R}$. Assuming the existence of $k_0$, condition \eqref{eq:separation} can be replaced by a much more relaxed condition $T<\min\cb{\tau_{p+1}-\tau_p}-\alpha$. This new condition ensures, via the hysteresis parameter $h$, that  $\tau_p\leq kT \leq \tau_{p+1}$. As shown numerically in the next section, the existence of $k_0$ takes place naturally in a practical scenario.

\begin{algorithm}[!t]
\label{alg:2}
\SetAlgoLined
{\bf Data:} $y_\eta\sqb{l},\psi_N\sqb{l},\lambda,h,\alpha,\Omega>0, l\in\cb{1,\dots,K}$.\\
\KwResult{ ${\mathrm{\widetilde{\gamma}}}\sqb{l},l\in\cb{1,\dots,K}$. }
 \begin{enumerate}[leftmargin =*, label = $\arabic*)$]
\item Compute ${y_{1,N}}\sqb{l}={\psi_N\ast y_\eta}\sqb{l}, k=1, p=0$.
\item While $k\leq K$ compute $\bar{y}_k={y_{k,N}}\sqb{k}$.
\begin{enumerate}[leftmargin =*, label = $2\alph*)$]
\item IF $k=n_p-N+2$: compute $\th=2\tfrac{\lambda_h}{N}$ \\
ELSE IF $k=n_p-N+3$: compute $\th=\tfrac{\lambda_h}{N}$\\
ELSE: $\th=\tfrac{\lambda_h}{2N}$.
\item IF $\vb{\bar{y}_k}>2\lambda_h-\th$: compute $p=p+1$ and\\
$\widetilde{n}_p=k+N, \widetilde{s}_p=-\mathrm{sign} \rb{\bar{y}_k}, \widetilde{\beta}_p=1$,\\
${y_{k+1,N}}\sqb{l}={y_{k,N}}\sqb{l}+2\lambda_h \widetilde{s}_p {d^{N-1}_{k+N}}\sqb{l}$, $k=k+1$.\\
ELSE IF $\th\leq\vb{\bar{y}_k}\leq2\lambda_h-\th$: compute $p=p+1$, ${y_{k+1,N}}\sqb{l}={y_{k,N}}\sqb{l}-\bar{y}_k {d^{N-1}_{k+N}}\sqb{l}$, and $\bar{y}_{k+1}=y_{k+1,N}\sqb{k+1}$.
\begin{itemize}[leftmargin =*]
\item IF $\vb{\bar{y}_{k}+\bar{y}_{k+1}-2\lambda_h \mathrm{sign} \rb{\bar{y}_k}} > \thbar$: compute\\
${y_{k,N}}\sqb{l}={y_{k-1,N}}\sqb{l}-\bar{y}_{k-1} {d^{N-1}_{k+N-1}}\sqb{l}$, \\
${y_{k+1,N}}\sqb{l}={y_{k,N}}\sqb{l}-\bar{y}_k {d^{N-1}_{k+N}}\sqb{l}$, \\
$\widetilde{n}_p=k+N-1, \widetilde{s}_p=-\mathrm{sign} \rb{\bar{y}_k},$ \\ $ \widetilde{\beta}_p=1-\tfrac{\vb{\bar{y}_k}}{2\lambda_h}$, $k=k+1$.\\
\item ELSE: compute\\
${y_{k+2,N}}\sqb{l}={y_{k+1,N}}\sqb{l}-\bar{y}_{k+1} {d^{N-1}_{k+N+1}}\sqb{l}$,  \\
$\widetilde{n}_p=k+N, \widetilde{s}_p=-\mathrm{sign} \rb{\bar{y}_k},$ $ \widetilde{\beta}_p=\tfrac{\vb{\bar{y}_k}}{2\lambda_h}$, $k=k+2$.\\
\end{itemize}
ELSE: Compute $k=k+1$
\end{enumerate}
\item For $p=1,\dots,P$ compute $\widetilde{\tau}_p=\widetilde{n}_p T-\alpha\cdot\widetilde{\beta}_p$.
\item Compute $\widetilde{\varepsilon}_\gamma\sqb{l}=\sum\nolimits _{p=1}^P s_p \varepsilon_0\rb{lT-\widetilde{\tau}_p}$, where  
$\varepsilon_0\rb{t}=2\lambda_h \sqb{ \rb{1/\alpha} t\cdot \ind_{[0,\alpha[}\rb{t}+ \ind_{[\alpha,\infty[}\rb{t}}$.
\item Compute ${\widetilde{\gamma}}\sqb{l} = {y_\eta}\sqb{l}+{\widetilde{\varepsilon}_\gamma}\sqb{l}$.
\end{enumerate}
\caption{Low Sampling Rate Recovery.}
\end{algorithm}

\begin{figure}[!t]
\centerline{\includegraphics[width=.45\textwidth]{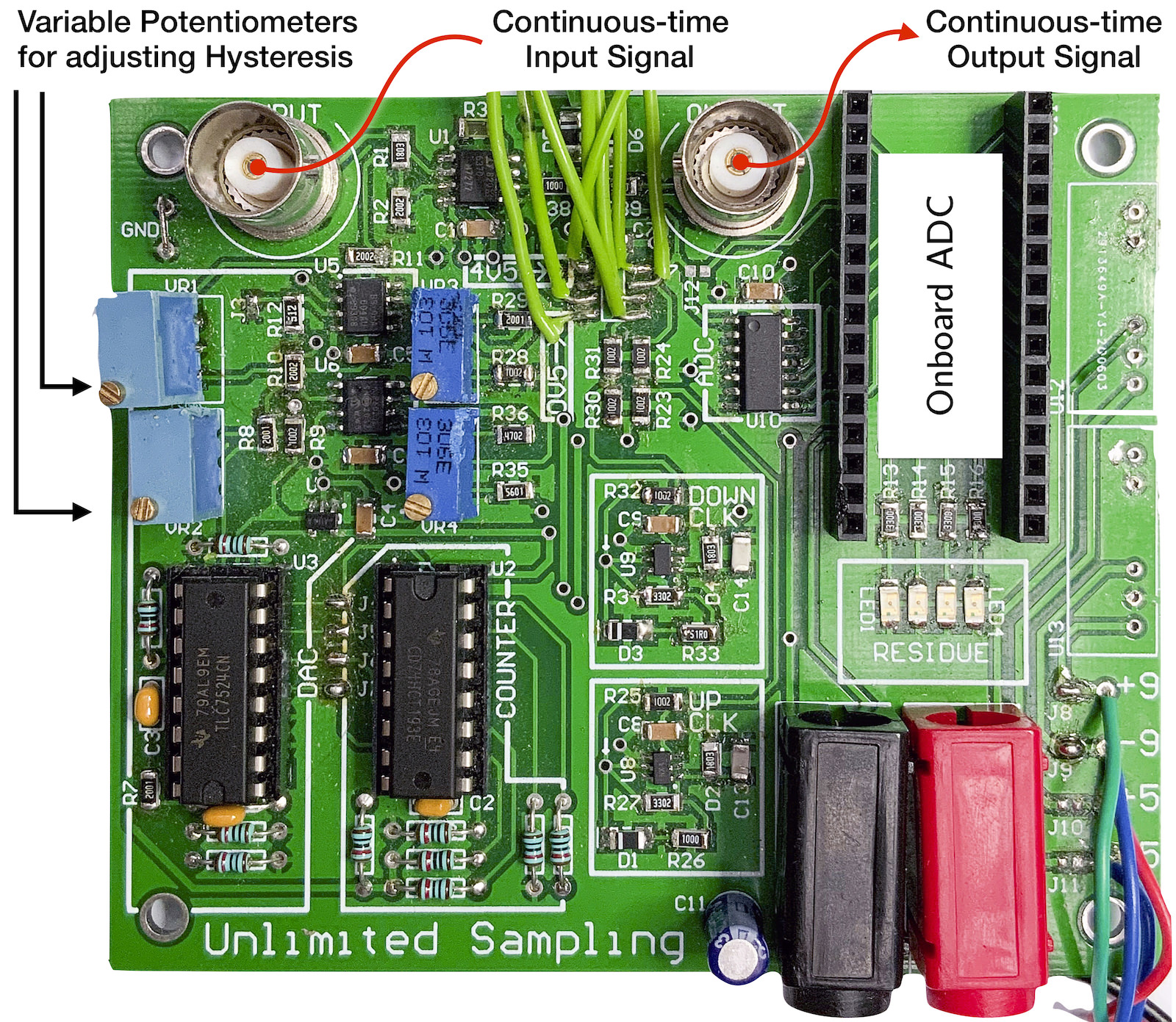}}
\caption{Modulo sampling hardware testbed for implementing hysteresis.}
\label{fig:USADC}
\end{figure}

\section{Numerical and Hardware Experiments}
\label{sect:numerical_demonstration}
{In what follows, we validate our recovery approach both for simulated data and data acquired from our modulo sampling hardware testbed (cf.~\fig{fig:USADC}). 
The experimental protocol for either case is discussed in Section~\ref{subsec:exp}. The hardware testbed is briefly discussed in Section~\ref{subsec:USADC}. We also compare our proposed methodology against $\mathsf{USAlg}$. The recovery results from synthetic data are {depicted} in \ref{sect:experiments_synthetic}, and the ones from experimental data are presented in \ref{sect:experiments_measured}. The algorithmic parameters and results are summarised in Table \ref{tab:1}.}

\subsubsection{Protocol for Experiments}
\label{subsec:exp}

In both cases, we consider a bandlimited function $g\in\PW{\Omega}$, process it with a modulo encoder and then compute uniform samples of the output obtaining $y\sqb{k}=\MOh g\rb{kT}$. We denote by $\widetilde{\gamma}_{\mathsf{USAlg}}=\mathsf{USAlg}_{\widetilde{\lambda}}\rb{y}$ the reconstruction via  \cite{Bhandari:2020a} from samples $y\sqb{k}$ using threshold $\widetilde{\lambda}$. When the data is encoded with an ideal modulo $\MO$ ($h=0, \alpha=0$), it makes sense to choose $\widetilde{\lambda}=\lambda$.  If $h\neq 0 $, then the hysteresis can be accounted for as shown in Subsection \ref{sect:hyst_USFrec} using $\widetilde{\lambda}=\lambda_h$. As we will see, however, in many cases where $\alpha,h\neq 0$ no choice of $\widetilde\lambda$ will lead to a satisfactory performance. To demonstrate this, we perform a line search to identify the effective threshold $\lambda_\mathsf{USAlg}$ for $\mathsf{USAlg}$ recovery 
\begin{equation}
\label{eq:lopt}
\lambda_\mathsf{USAlg}=\mathrm{argmin}_{\widetilde{\lambda}} \ \mathsf{MSE}\rb{\gamma,\mathsf{USAlg}_{\widetilde{\lambda}}\rb{y}}.
\end{equation}
Furthermore, we denote by $\widetilde{\gamma}_\mathsf{TH}$ the reconstruction with the proposed thresholding algorithm outlined in \ref{alg:1}.
We measure the relative recovery error of $\mathsf{USAlg}$ and the proposed thresholding method via the ratio $\mathsf{Err}\rb{\widetilde{\gamma},\gamma}$, defined as
\begin{equation}
	\mathsf{Err}\rb{\widetilde{\gamma},\gamma}=100\cdot\frac{\mathsf{MSE}\rb{\widetilde{\gamma},\gamma}}{\mathsf{MSE}\rb{\gamma,0}} \rb{\%}.
\end{equation}
{We denote the input recovery errors with each method by $\mathsf{Err}_{\mathsf{TH}}=\mathsf{Err}\rb{\widetilde{\gamma}_\mathsf{TH},\gamma}$ and $\mathsf{Err}_{\mathsf{USAlg}}=\mathsf{Err}\rb{\widetilde{\gamma}_\mathsf{USAlg},\gamma}$, respectively.}
When the number of folding times is identified correctly, we compute the root mean square error (RMSE) between the correct and estimated folding times as $\mathsf{RMSE}\rb{\widetilde{\tau},\tau}=\sqrt{\mathsf{MSE}\rb{\widetilde{\tau},\tau}}$.
To compare between the coarse and fine estimations of the folding times for the proposed thresholding method, in this case we also compute the coarse folding times estimation error as $\mathsf{RMSE}\rb{\widetilde{n}T,\tau}$.

 {
\subsubsection{Hardware Testbed}
\label{subsec:USADC}

To perform experiments with real signals, we have developed a modulo sampling hardware that can be implemented with off-the-shelf components. The circuit  is shown in \fig{fig:USADC}. The basic design follows from \cite{Bhandari:2021:J} but with a key difference; the hardwired variable potentiometers (annotated in \fig{fig:USADC}) allow for the adjustment of the hysteresis parameter $h$ in Definition~\ref{def:cont_modulo}. 
This is done by adjusting the multi-turn screw of the trimmer component. To estimate the values of $\lambda, h$ and $\alpha$ from the data, we first recorded the output of the modulo ADC with a high sampling rate and then ran a nonlinear optimization to fit the values of $\lambda$ and $h$. For $\alpha$, we estimated the transient slope directly from the data. For a given continuous-time input signal, continuously adjusting the multi-turn value amounts to changing the $h$-value in the output signal. This is what is shown in the {YouTube} demonstration at \href{https://youtu.be/R4rji5jOjD8}{\texttt{https://youtu.be/R4rji5jOjD8}}. While the details of the hardware design are beyond the scope of this paper, in the spirit of reproducible research,  we plan to make our hardware designs, data and code available via the project website  \href{https://bit.ly/USF-RR}{\texttt{https://bit.ly/USF-RR}} in the future. We stress, however, that our method does not require variable hysteresis to work, but only to avoid $h\approx 0$. We change $h$ with the help of the testbed in the hardware experiments below purely to show the diversity of cases in which the algorithm works.}

\begin{table}[!t]
\caption{Summary of the Parameters and the Performance of $\mathsf{USAlg}$ and the Proposed Thresholding Algorithm.}
\label{tab:1}
\centering
\resizebox{\textwidth}{!}{%
\begin{tabular}{@{}ccccccccccccc@{}}
\toprule
Exp. & $T\rb{\mathrm{ms}}$ & $\Omega\rb{\mathrm{rad/s}}$ & $\lambda$ & $h$    & $\alpha\rb{\mathrm{ms}}$ & $\lambda_\mathsf{USAlg}$ & $\const{N}_\mathsf{USAlg}$ & $\tfrac{\lambda_h}{2\const{N}_\mathsf{TH}}$ & $\const{N}_\mathsf{TH}$ & $\mathsf{Err}_{\mathsf{USAlg}}\rb{\%}$ & $\mathsf{Err}_{\mathsf{TH}}\rb{\%}$ & $\mathsf{RMSE}\rb{\widetilde{\tau}^\const{N}_\mathsf{TH},\tau}$ \\ \midrule
1    & $20$                & $4.4$                       & $1.5$     & $1.5$  & $20$                     & $0.66$                   & $1$                        & $0.125$                                     & $3$                     & $25.6$                                 & $0.008 $                            & $1.2\times10^{-5}\ \mathrm{s}$                                  \\
     &                     &                             &           &        &                          &                          & $2$                        & $0.094$                                     & $4$                     & $8.2\times10^{3}$                      & $4.5\times10^{-4} $                 & $6.5\times10^{-7}\ \mathrm{s}$                                  \\ \midrule
2    & $0.1$               & $188$                       & $2.01$    & $3.23$ & $0.09$                   & $0.13$                   & $1$                        & $0.198$                                     & $1$                     & $2.64 $                                & $0.58 $                             & $0.03\ \mathrm{ms}$                                             \\
     &                     &                             &           &        &                          &                          & $2$                        & $0.098$                                     & $2$                     & $730 $                                 & $0.33 $                             & $0.01\ \mathrm{ms}$                                             \\ \midrule
3    & $0.36$              & $188$                       & $2.05$    & $1$    & $0.07$                   & $0.78$                   & $2$                        & $0.39$                                      & $2$                     & $27 $                                  & $1 $                                & $0.32\ \mathrm{ms}$                                             \\
4    & $0.09$              & $1.5$                       & $1$       & $0.1$  & $0.02$                   & $1.345$                  & $1$                        & $0.24$                                      & $2$                     & $96.8$                                 & $0.71$                              & $0.04\ \mathrm{ms}$                                             \\ \bottomrule
\end{tabular}%
}
\end{table}

\begin{figure*}[!t]
\centering
\subfloat{\includegraphics[width=.32\textwidth]{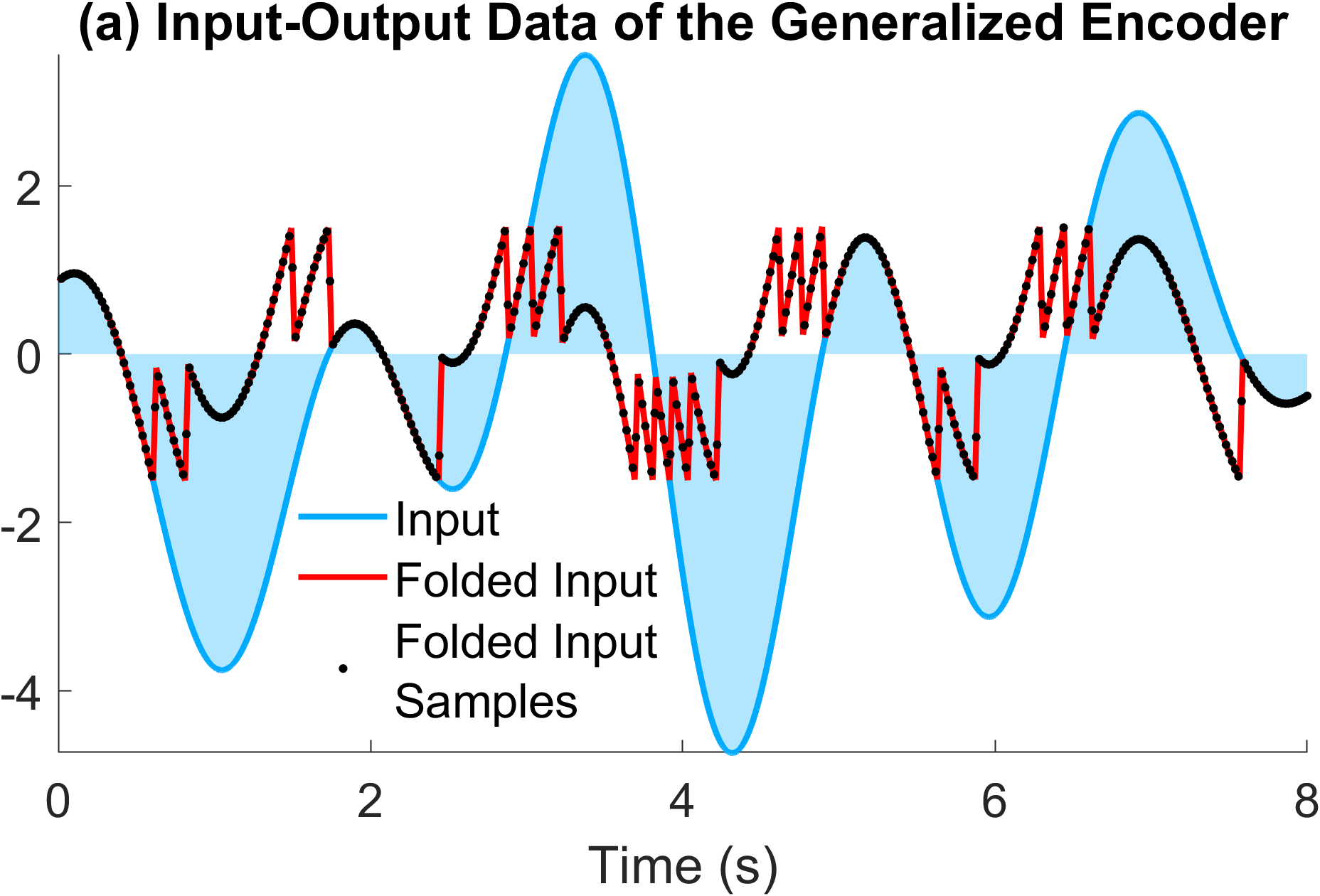}
\label{fig:synth_data}}
\subfloat{\includegraphics[width=.33\textwidth]{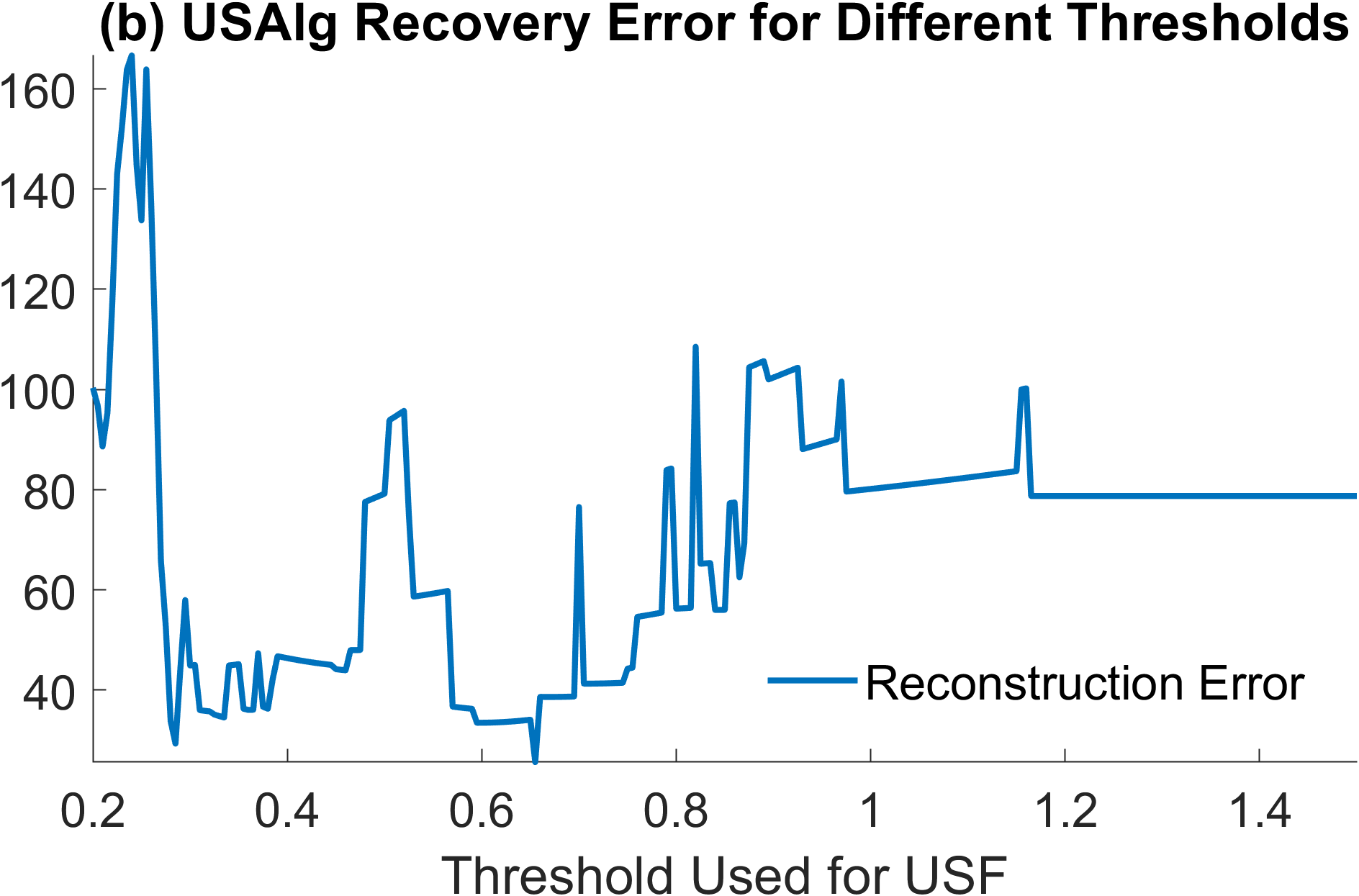}
\label{fig:synth_USF_Lam}}
\subfloat{\includegraphics[width=.33\textwidth]{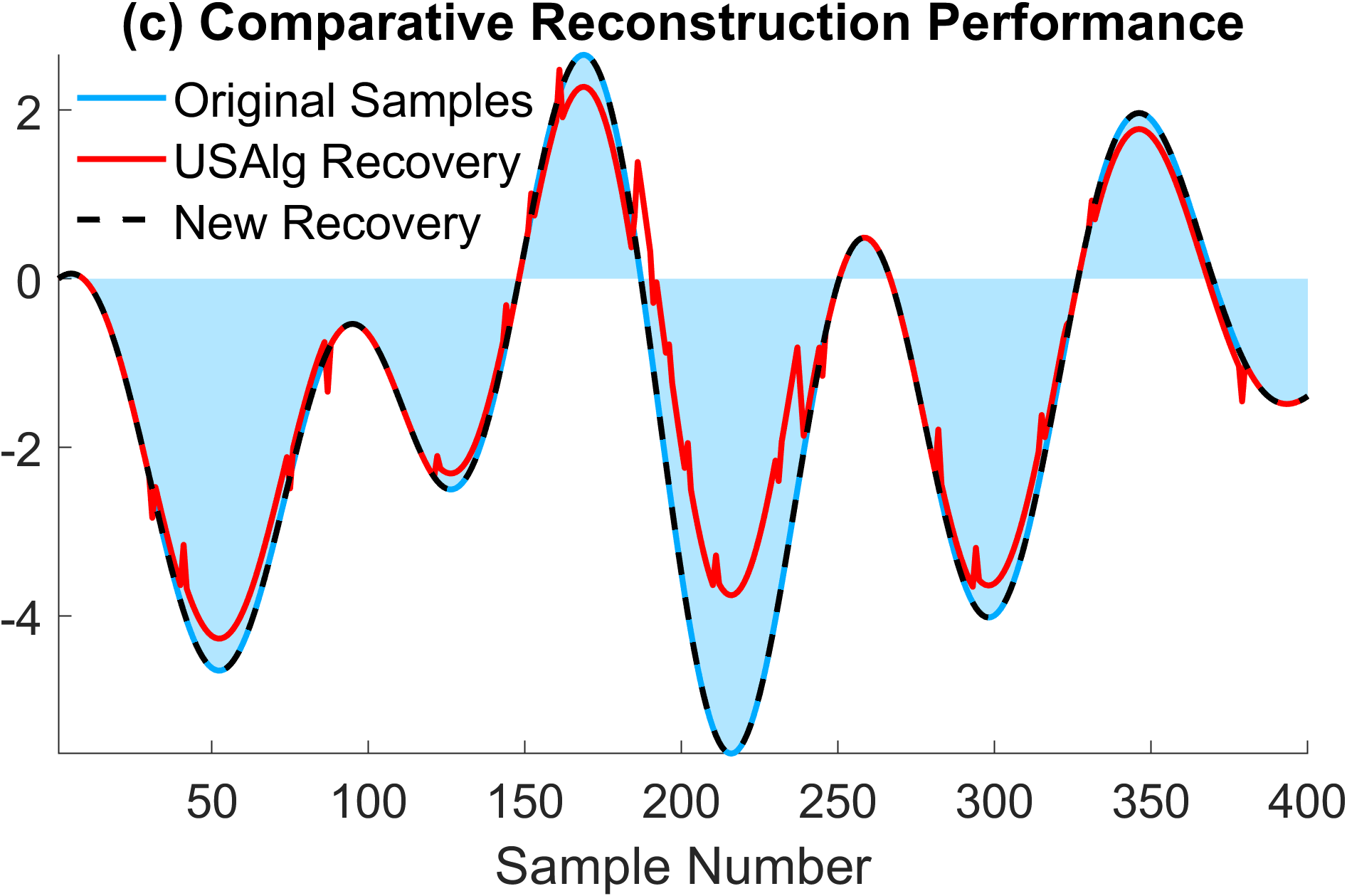}
\label{fig:synth_rec}}
\hfil
\subfloat{\includegraphics[width=.315\textwidth]{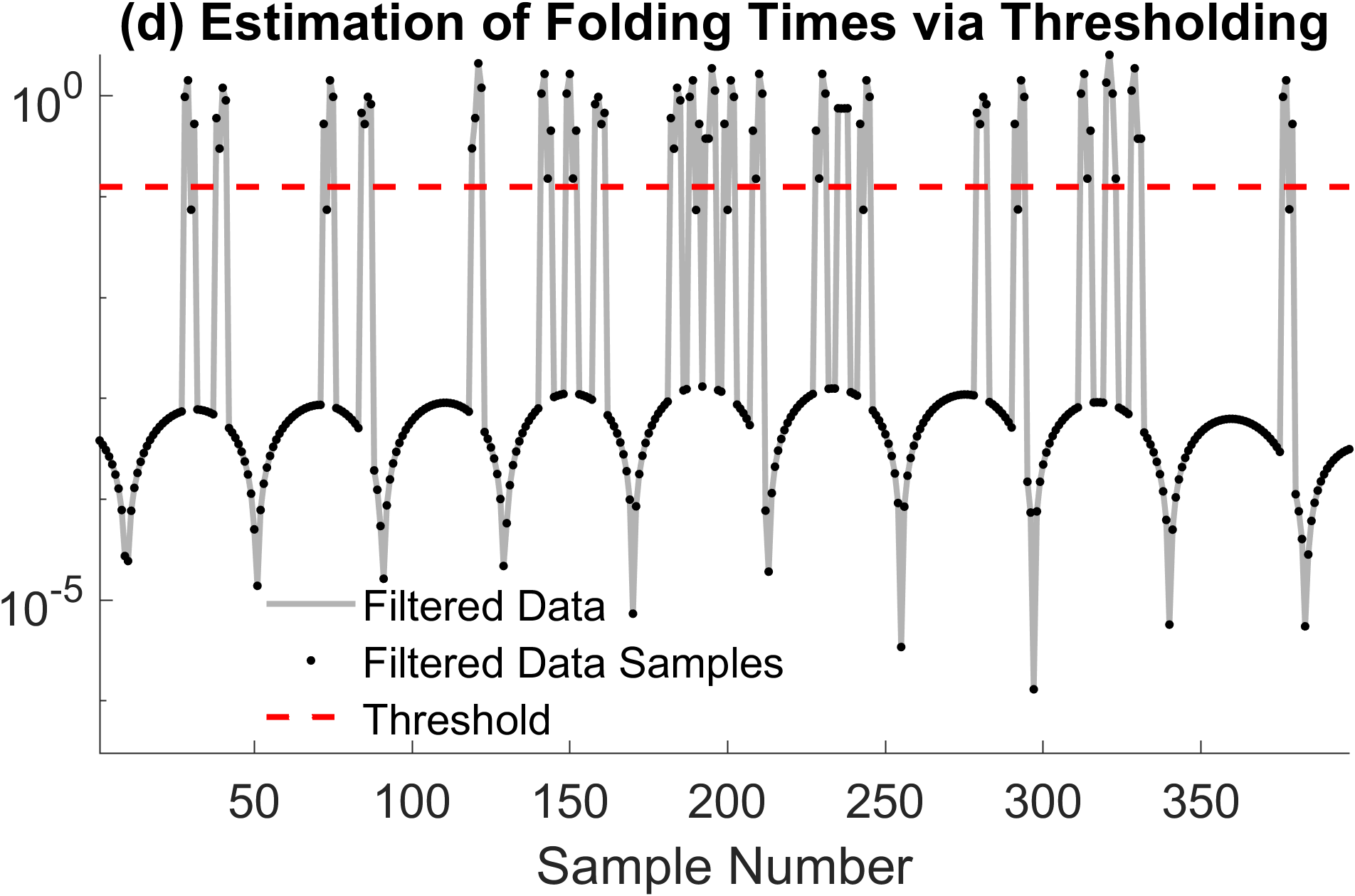}
\label{fig:synth_th}}
\subfloat{\includegraphics[width=.32\textwidth]{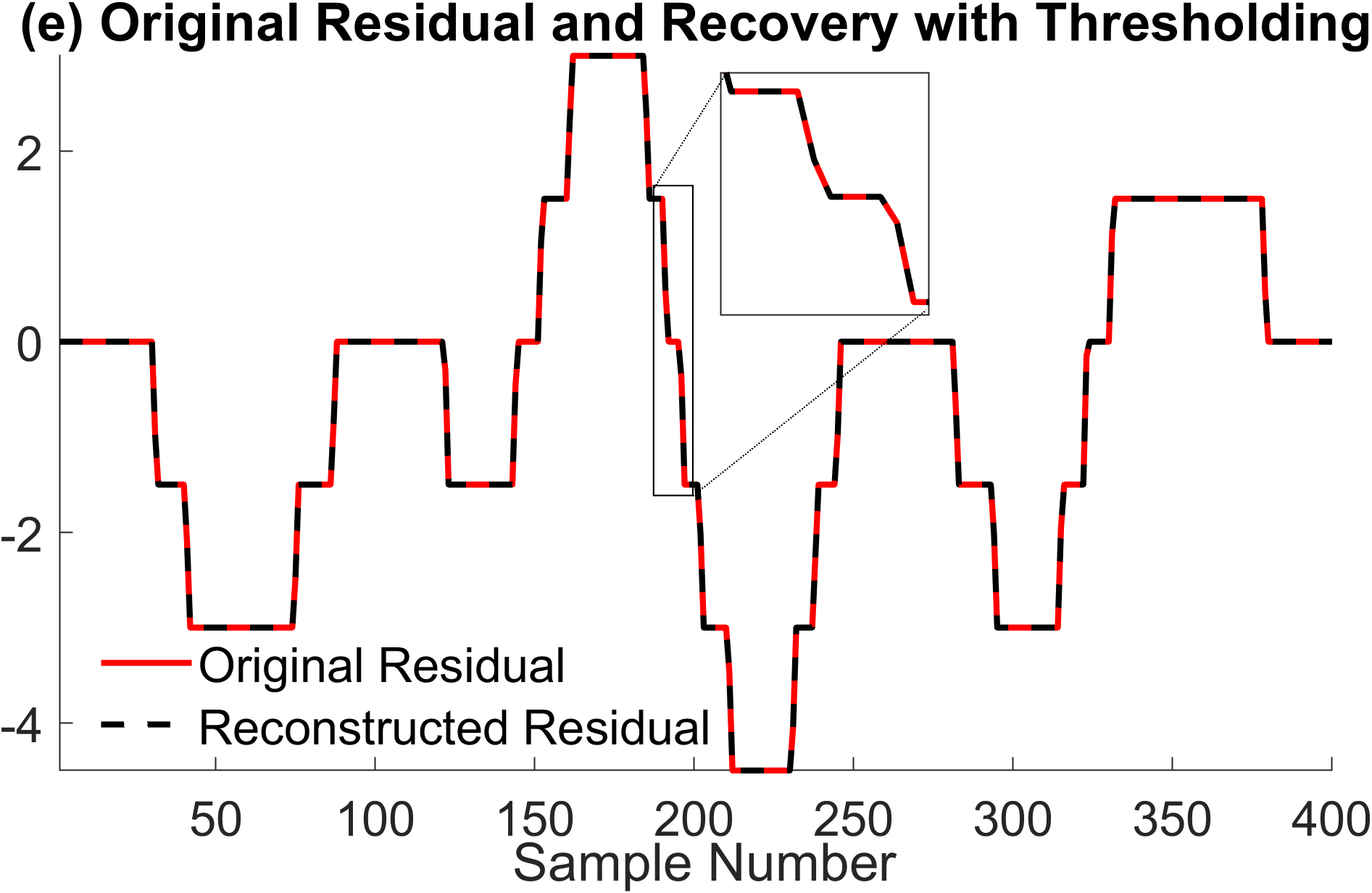}
\label{fig:synth_res}}
\subfloat{\includegraphics[width=.345\textwidth]{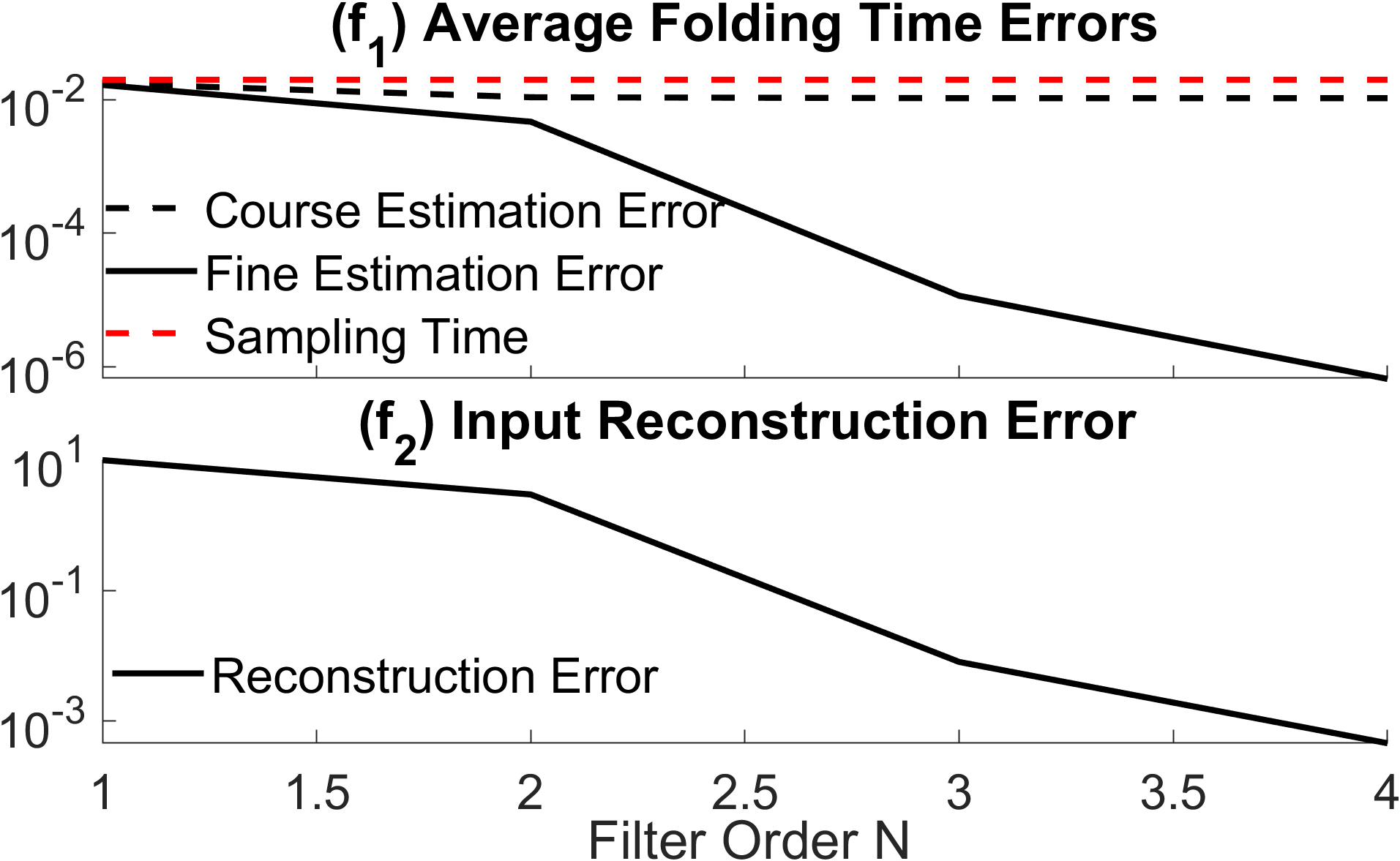}
\label{fig:synth_fold}}
\caption{Simulations with synthetic data for $T=\alpha$. (a) A randomly generated signal and its modulo samples. (b) Reconstruction error via $\mathsf{USAlg}$ for various values of the effective threshold. (c) Reconstruction with $\mathsf{USAlg}$ using the optimal effective threshold in comparison with the reconstruction by thresholding (d) Filtered data
and the resulting estimated folding times via thresholding. 
(e) The residual function recovered with the proposed method; note that the samples inside the transient periods are recovered correctly. (f) Reconstruction error as a function of the filter order $\const{N}$. (f\textsubscript{1}) The folding time errors based on the coarse and fine estimation. 
(f\textsubscript{2}) Resulting reconstruction accuracy calculated over the uniform sample grid.}
\label{fig_sim}
\end{figure*}

\subsection{Reconstruction from Synthetic Data}	
\label{sect:experiments_synthetic}

In this class of experiments, we choose the input $g\rb{t}$ to be a  {non-periodic signal computed as a sum of $\mathrm{sinc}$ functions bandlimited to $\Omega=4.4\ \mathrm{rad/s}$, centered in $10$ points equally spaced by $\pi/\Omega$ in $\sqb{0,8}$. The $\mathrm{sinc}$ coefficients were randomly generated using the uniform distribution on $\sqb{-4\lambda,4\lambda}$, where $\lambda$ is the modulo threshold satisfying $\lambda=1.5$.} The modulo hysteresis parameter is $h=1.5$, and transient length is $\alpha=0.02\ \mathrm{s}$. The modulo output is subsequently sampled with period $T=0.02\ \mathrm{s}$. An example input and the corresponding output of the encoder are depicted in \fig{fig:synth_data}. We attempt to recover the input using $\mathsf{USAlg}$. Let $\widetilde{\gamma}_\mathsf{USAlg}=\mathsf{USAlg}_{\lambda_\mathsf{USAlg}}\rb{y}$. We first choose $\const{N}=1$. 
The value of $\mathsf{Err}_{\mathsf{USAlg}}$ for $200$ equally spaced values $\lambda_\mathsf{USAlg}\in\sqb{0.2,1.5}$ is depicted in \fig{fig:synth_USF_Lam}. The plot confirms that indeed $\lambda_\mathsf{USAlg}$, the optimal choice of the threshold parameter $\widetilde \lambda$ in the reconstruction procedure, satisfies $\lambda_\mathsf{USAlg} \notin \{\lambda_h, \lambda\}$.

The reconstruction with $\lambda_\mathsf{USAlg}\approx 0.655$ is depicted in \fig{fig:synth_rec} superimposed on the original samples. The error is evaluated as $\mathsf{Err}_{\mathsf{USAlg}}=25.6\%$. As it can be seen, $\mathsf{USAlg}$ unfolds the samples incorrectly, which is because $\alpha\neq 0$. Even for the optimal value $\lambda_\mathsf{USAlg}$, increasing $\const{N}$ renders the reconstruction unstable, leading to $\mathsf{Err}_{\mathsf{USAlg}}>1000\%$.  We also note that the estimation is more accurate where the samples are close to the edges of the transient periods in \fig{fig:synth_data}, e.g., the 5\textsuperscript{th} and 6\textsuperscript{th} folding time. While changing $T$ may reduce the number of samples during the transient periods, it will not be possible to know the optimal values of $T$ and $\lambda_\mathsf{USAlg}$, as they depend on the input.  {We note that condition \eqref{eq:gamma_bound_hyst} is satisfied for $N=2$, which would be necessary in the hypothetical scenario where no transients were present ($\alpha=0$).} 

We reconstruct the input samples using the proposed method defined in Section \ref{sect:thresholding}. We choose $\const{N}=3$, and subsequently threshold the filtered data samples $\Delta^3 y$ with ${\lambda_h}/{6}=0.125$. As illustrated in \fig{fig:synth_th}, the samples corresponding to the folding times are clustered in groups of $4$ consecutive samples. As Lemma \ref{lem:k1_k2} demonstrates, each cluster has a minimum of $3$ samples above the threshold, which are enough to compute $\widetilde{n}_p, \widetilde{s}_p, \widetilde{\tau}_p, p=1,\dots,22$ sequentially via Theorem \ref{theo:s1_tau1}. The number of folds is unknown prior to the reconstruction. The reconstruction of the residual function $\widetilde{\varepsilon}_\gamma\sqb{k}$ and of the input samples $\widetilde{\gamma}_\mathsf{TH}\sqb{k}$ are depicted in \fig{fig:synth_res} and \fig{fig:synth_rec}, respectively. We note that the proposed methodology accurately recovers all the samples, including those on the transient regions of $\varepsilon_\gamma\sqb{k}$. The error is evaluated as $\mathsf{Err}_{\mathsf{TH}}=8.1\cdot 10^{-3}\%$. We compute a reconstruction $\widetilde{\gamma}_{\mathsf{TH},\const{N}}$ for each $\const{N}=1,\dots,4$ and observe that even though $\left.{\mathrm{TH}}_2\right)$ is not satisfied for $\const{N}=4$, we an obtain accurate recovery. 
\begin{figure*}[t!]
\centering
\subfloat{\includegraphics[width=.5\textwidth]{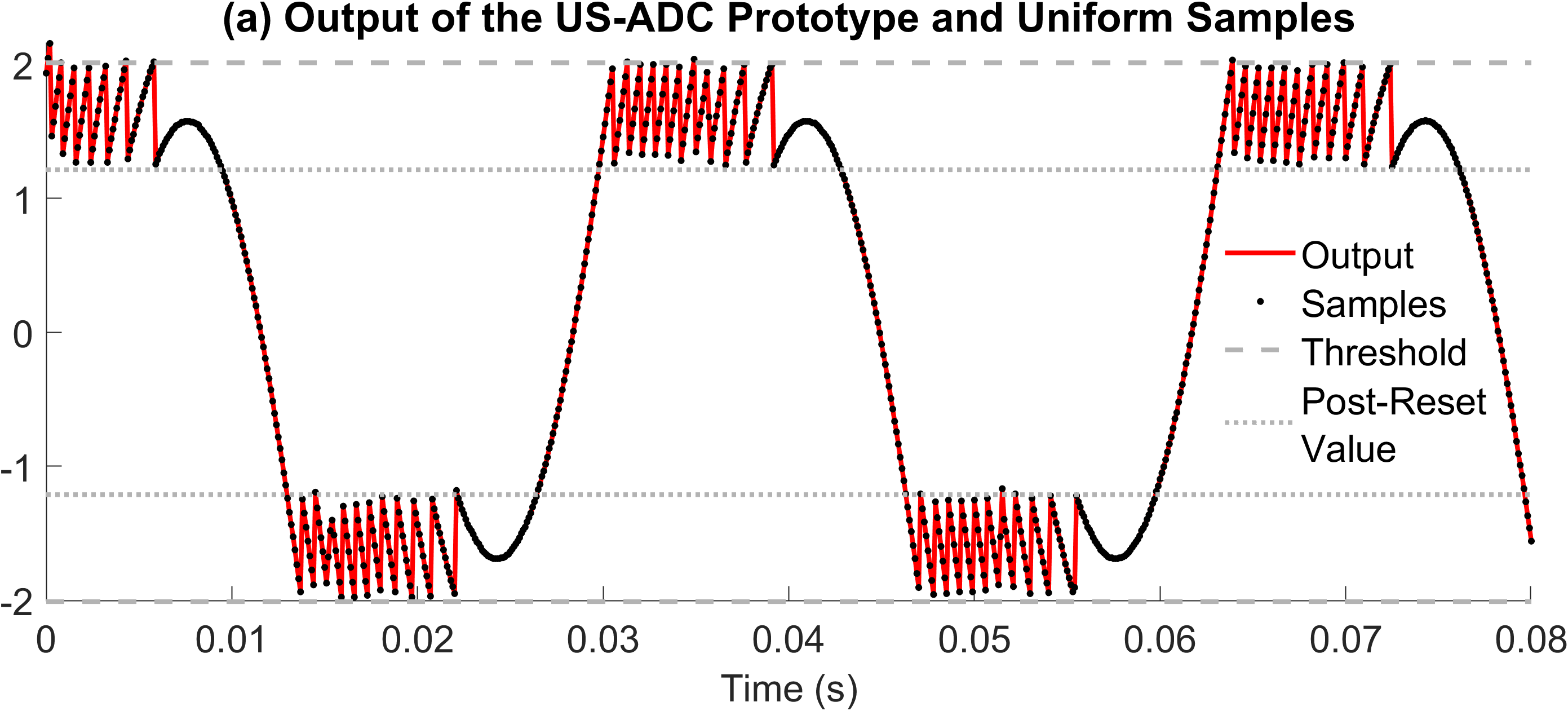}
\label{fig:measured_samples}}
\subfloat{\includegraphics[width=.5\textwidth]{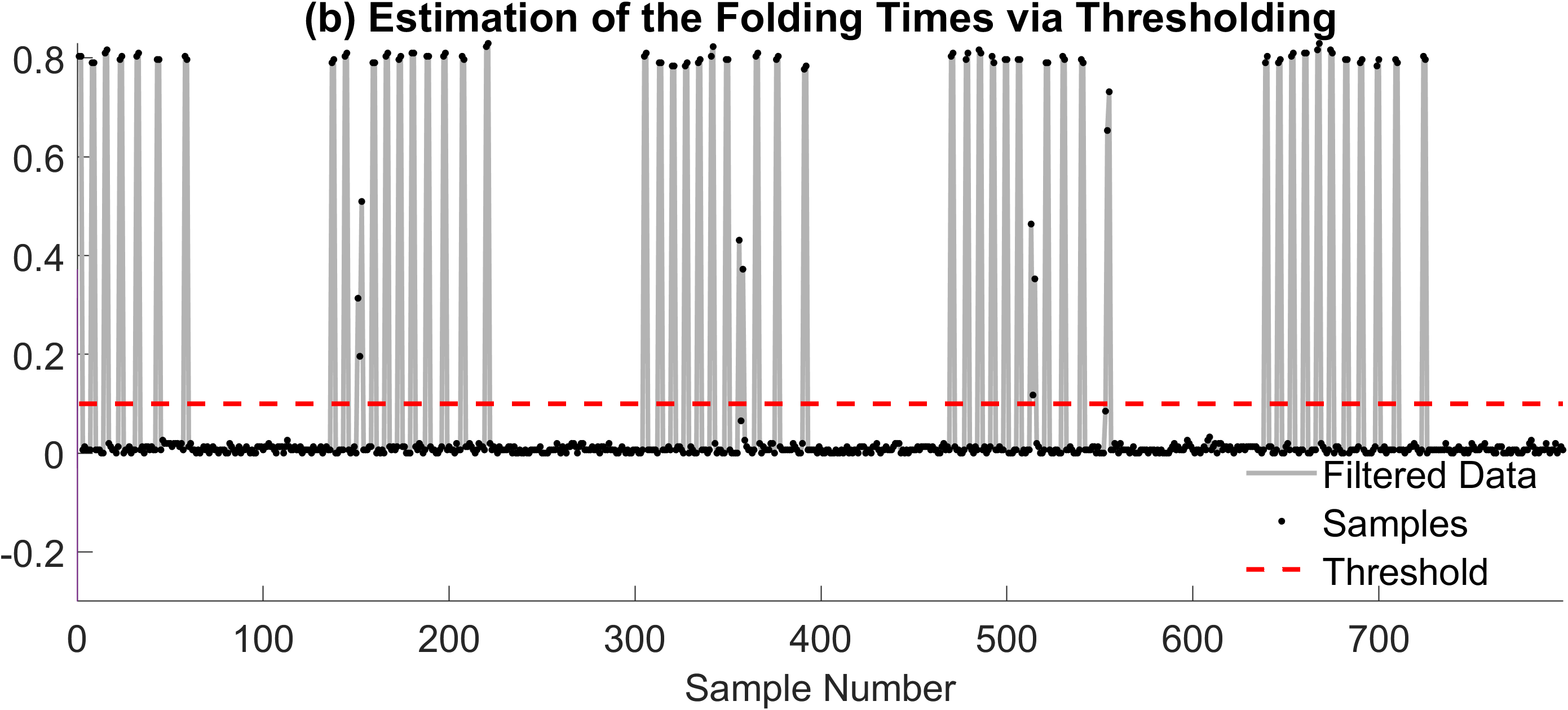}
\label{fig:measured_th}}
\hfil
\subfloat{\includegraphics[width=.5\textwidth]{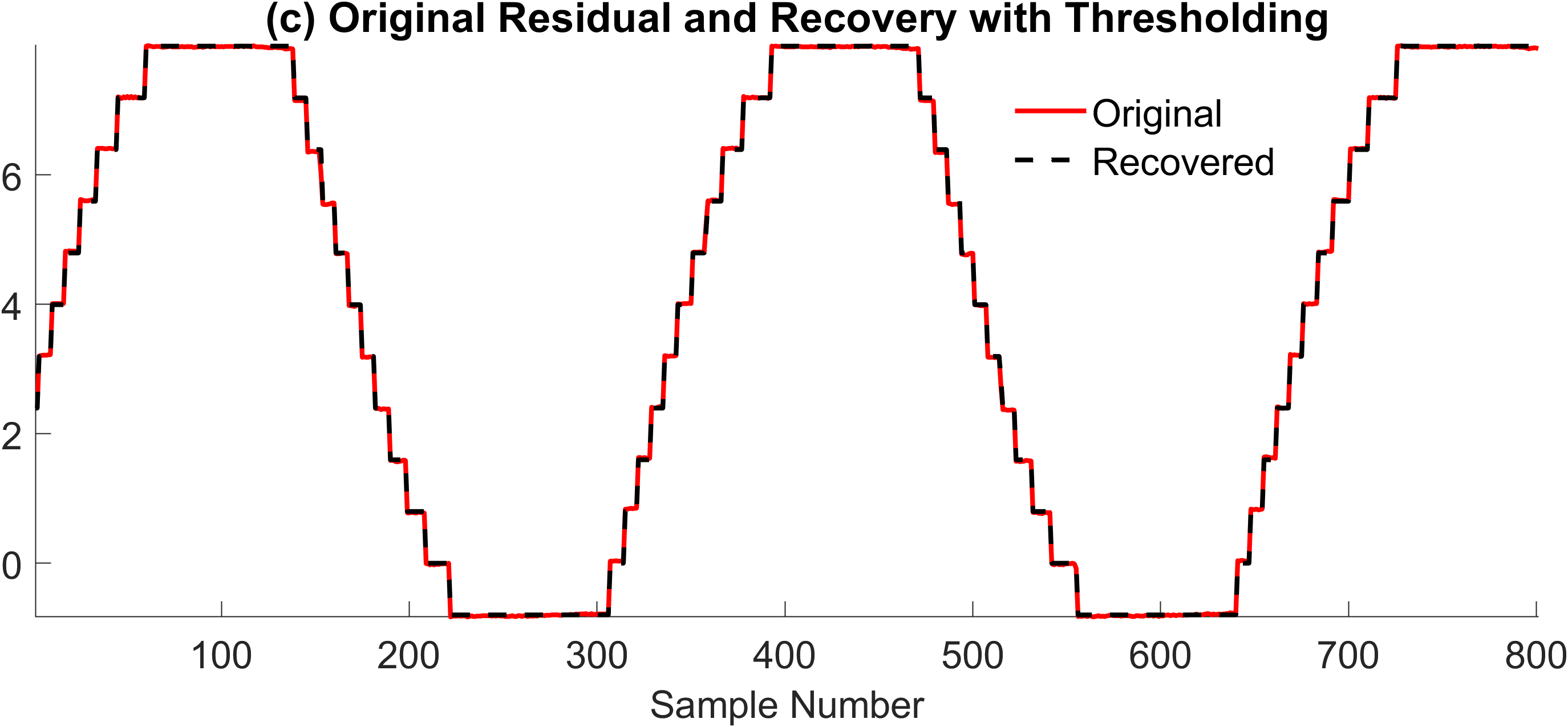}
\label{fig:measured_res}}
\subfloat{\includegraphics[width=.5\textwidth]{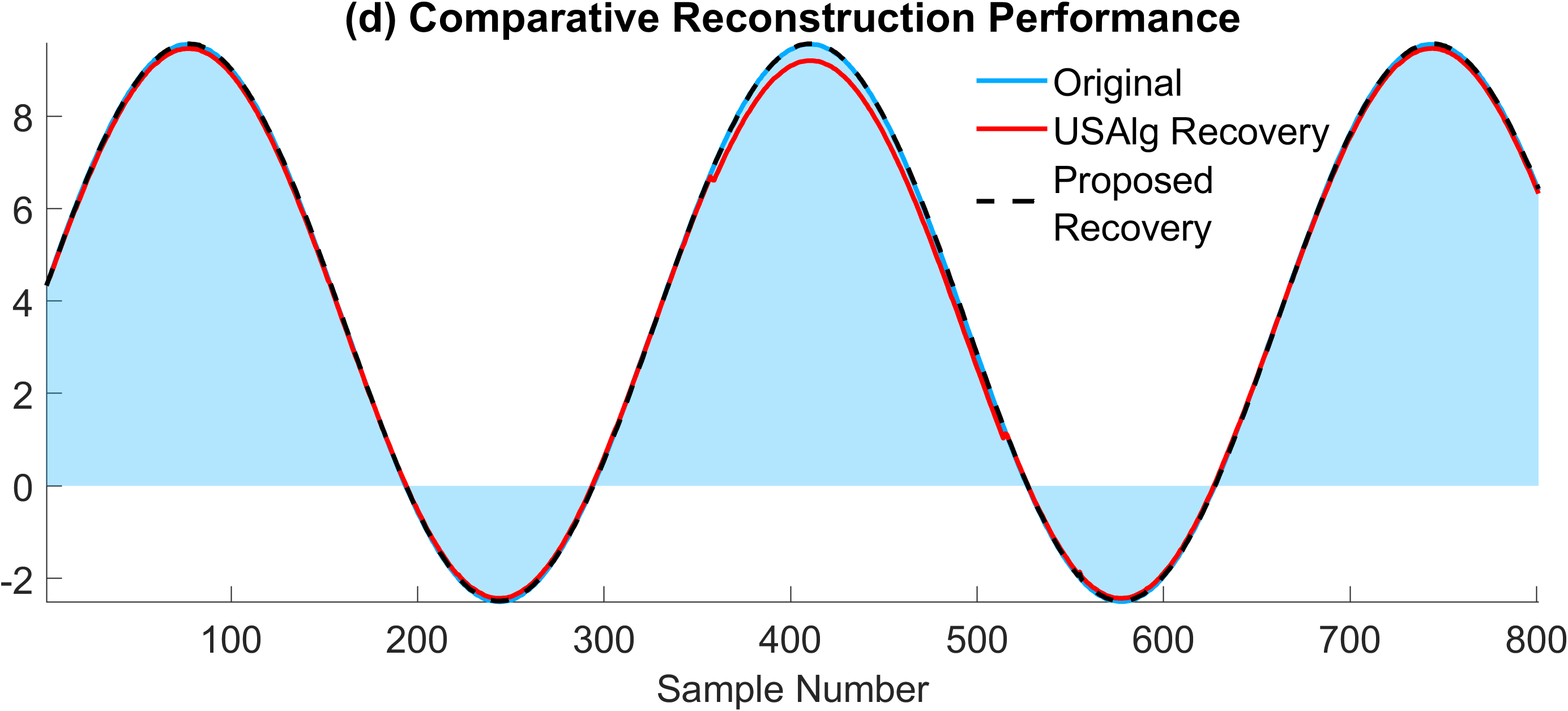}
\label{fig:measured_rec}}
\caption{Hardware experiment (Sinusoid). (a) Modulo circuit output and corresponding uniform samples. (b) Estimating the folding times by thresholding the filtered data. (c) Reconstruction of the residual function. (d) Input samples and reconstruction with $\mathsf{USAlg}$ with effective threshold and the proposed method. }
\label{fig_measured1}
\end{figure*}
An explanation for this observation is that as $\const{N}$ increases, the size of the clusters in \fig{fig:synth_th} increases linearly, which allows relaxing $\left.{\mathrm{TH}}_2\right)$ right until the clusters begin to overlap. This perspective also can serve as a guideline to choose $N$ in practice; we aim to explore details in future work.

For each $\const{N}$, we evaluate the folding time estimation, $\mathsf{RMSE}\rb{\widetilde{\tau}^\const{N},\tau}$. \fig{fig:synth_fold} shows the evolution of $\mathsf{RMSE}\rb{\widetilde{\tau}^\const{N},\tau}$ and $\mathsf{Err}\rb{\widetilde{\gamma}_{\mathsf{TH},\const{N}},\gamma}$ for $\const{N}\in[1,4]$.
The coarse estimations or CEs of the folding times, denoted $\widetilde{n}_p^N T$, are also evaluated for comparison purposes. As expected, CE leads to an error $\mathsf{RMSE}\rb{\widetilde{n}_p^\const{N} T,\tau_p}$ comparable to the sampling time $T$, and does not change significantly with $\const{N}$. However, the fine estimation \eqref{eq:folding_seq} allows for computing $\tau_p$ with an error that is $5$ orders of magnitude smaller than $T$. This is remarkable given that no information is being transmitted at times $\tau_p$.

\subsection{Reconstruction from Experimental Data}
\label{sect:experiments_measured}

In this example, we consider data generated by our prototype ADC with modulo threshold $2.01$ and hysteresis parameter $h=3.23$ applied to an input $g\rb{t}$ given as a sine wave of amplitude $12$ and bandwidth $\Omega=188\ \mathrm{rad/s}$. 
The output of the circuit is depicted in \fig{fig:measured_data}. 
The output $z\rb{t}$ is sampled with period $T=10^{-4}$. 
The samples, depicted in \fig{fig:measured_samples}, show that the experimental data behaves very much in line with our model in Definition \ref{def:cont_modulo}. When the encoder input reaches one of the threshold values $\pm\lambda$, it resets the output to the corresponding post-reset value $\mp\rb{\lambda-h}$.

As before, we reconstruct the input with $\mathsf{USAlg}$ starting with $\const{N}=1$ by selecting the effective threshold using a line search based on $200$ equally spaced values $\lambda_\mathsf{USAlg}\in\sqb{0,2}$. The smallest $\mathsf{Err}_{\mathsf{USAlg}}=2.6\%$ was achieved for $\lambda_\mathsf{USAlg}=0.13$. The corresponding reconstruction is depicted in \fig{fig:measured_rec}. Increasing $N$ further led to an unstable recovery. We note that condition \eqref{eq:gamma_bound_hyst}, which assumes $\alpha=0$, is satisfied for $N=2$.
The reconstruction with the proposed method satisfies the requirements $\left.{\mathrm{TH}}_1\right)$ $\&$ $\left.{\mathrm{TH}}_2\right)$ for $\const{N}=1$. The filtered data was thresholded with $\frac{\lambda_h}{2}=0.198$, ensuring that each folding time corresponds to a minimum of $1$ sample crossing the threshold \fig{fig:measured_th}. The recovered residual is in \fig{fig:measured_res}. The corresponding input reconstruction satisfies $\mathsf{Err}_{\mathsf{TH}}=0.58\%$ and is depicted in \fig{fig:measured_rec}. Even higher accuracy can be achieved with a filter order of $\const{N}=2$ which brings the error to $\mathsf{Err}_{\mathsf{TH}}=0.33\%$.

In a second experiment with our prototype ADC, we generate a random signal with bandwidth $188\ \mathrm{rad/s}$. The sampling rate is $T=0.36\ \mathrm{ms}$, and the circuit parameters are $\lambda=2.05, h=1$. We select $N=2$, which is also the value predicted by $\mathsf{USAlg}$ using the effective threshold $\widetilde{\lambda}=\lambda_h$. A line search on $\sqb{0,2}$ yields $\lambda_\mathsf{USAlg}=0.78$ as the optimal threshold of $\mathsf{USAlg}$, for which the error is given by $\mathsf{Err}_{\mathsf{USAlg}}=27\%$. We note that condition \eqref{eq:gamma_bound_hyst}, which assumes $\alpha=0$ is satisfied for $N=2$. The input and $\mathsf{USAlg}$ reconstruction are depicted in \fig{fig:measured_rec2}. For the proposed method, the reconstruction condition is satisfied for $N=2$, which leads to a recovery threshold of ${\lambda_h}/{4}=0.39$, which leads to error $\mathsf{Err}_{\mathsf{TH}}=1\%$. The proposed reconstruction method is depicted in \fig{fig:measured_th2}, \fig{fig:measured_res2}, and \fig{fig:measured_rec2}.  {The figures illustrate how in an experimental setting the proposed method allows the reconstruction of signals with the dynamic range up to $10$ times the modulo threshold.}

\begin{figure*}[t!]
\centering
\subfloat{\includegraphics[width=.46\textwidth]{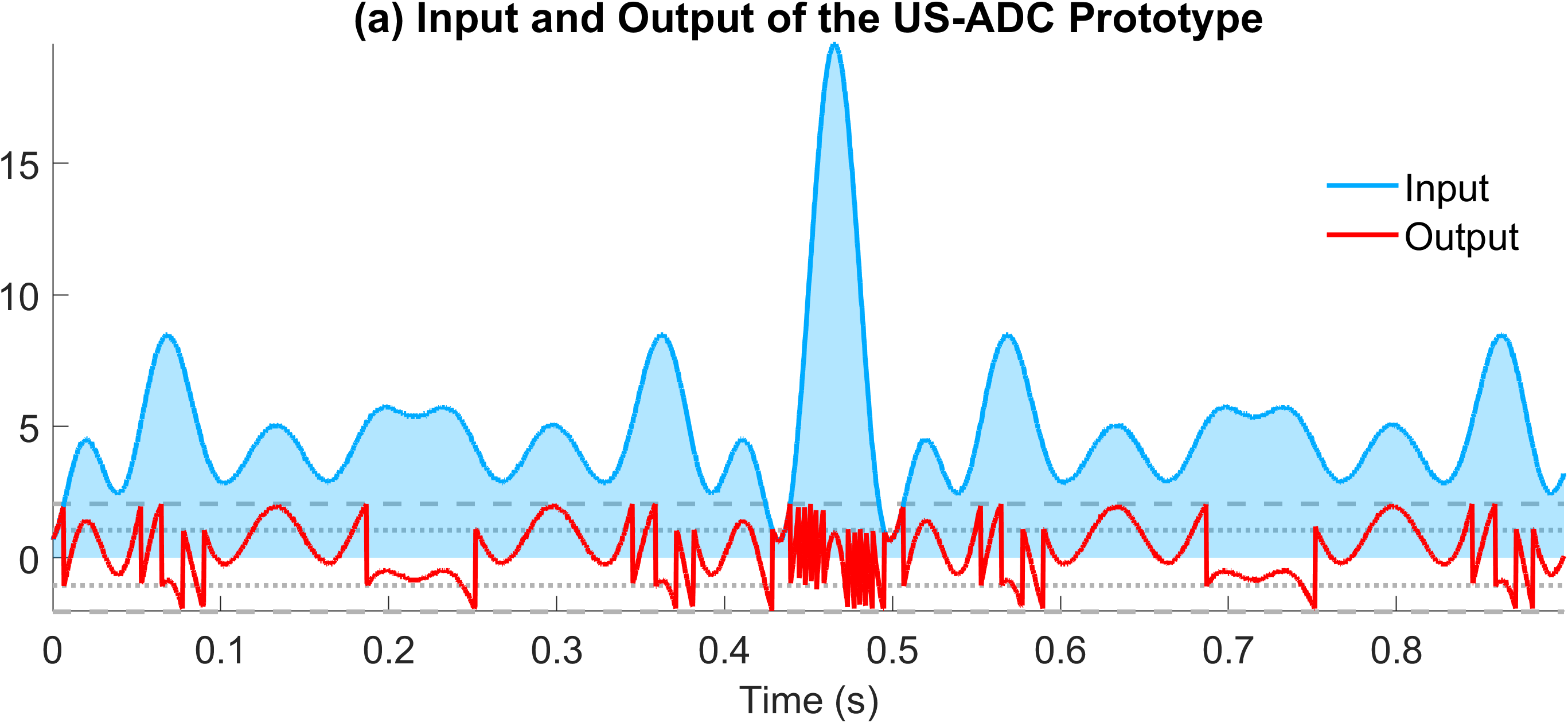}
\label{fig:measured_samples2}}
\subfloat{\includegraphics[width=.46\textwidth]{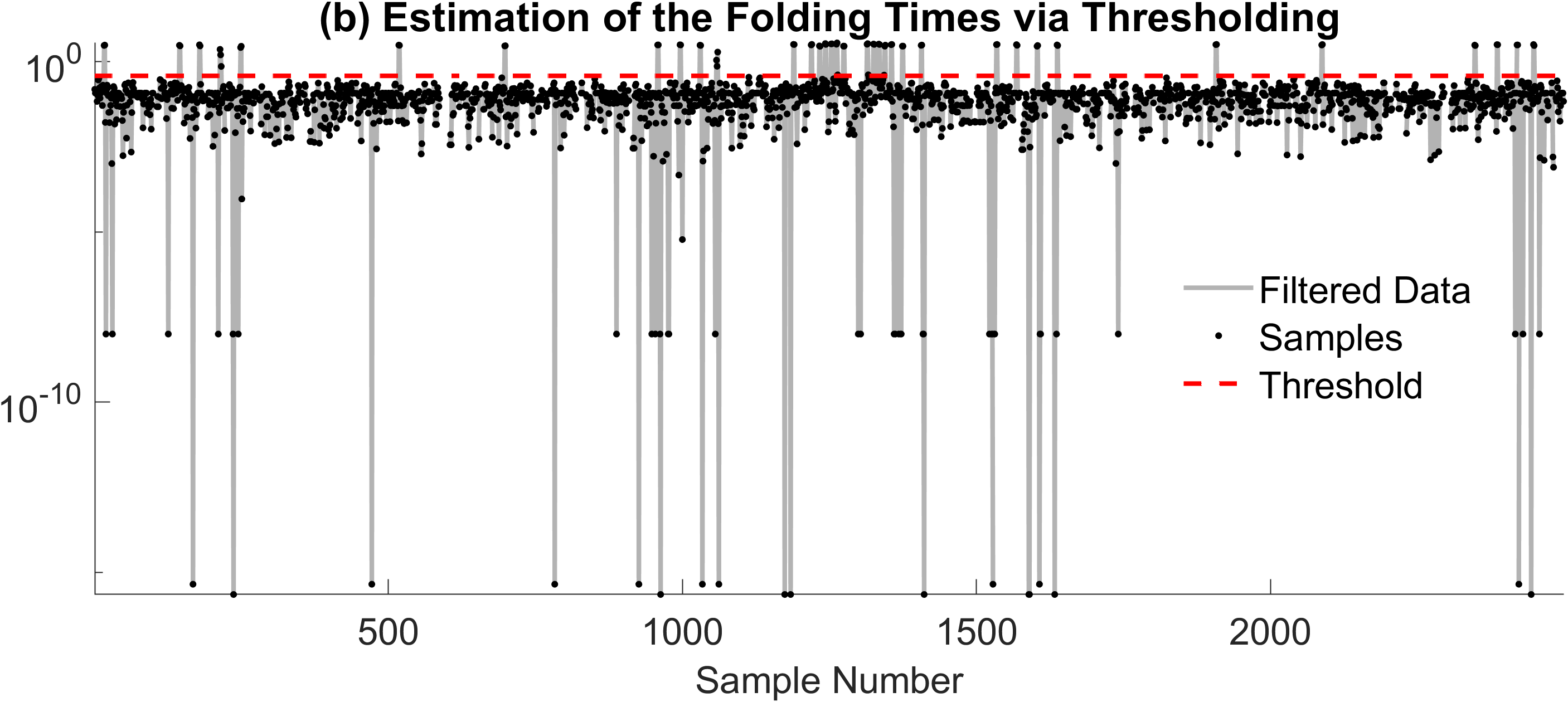}
\label{fig:measured_th2}}
\hfil
\subfloat{\includegraphics[width=.46\textwidth]{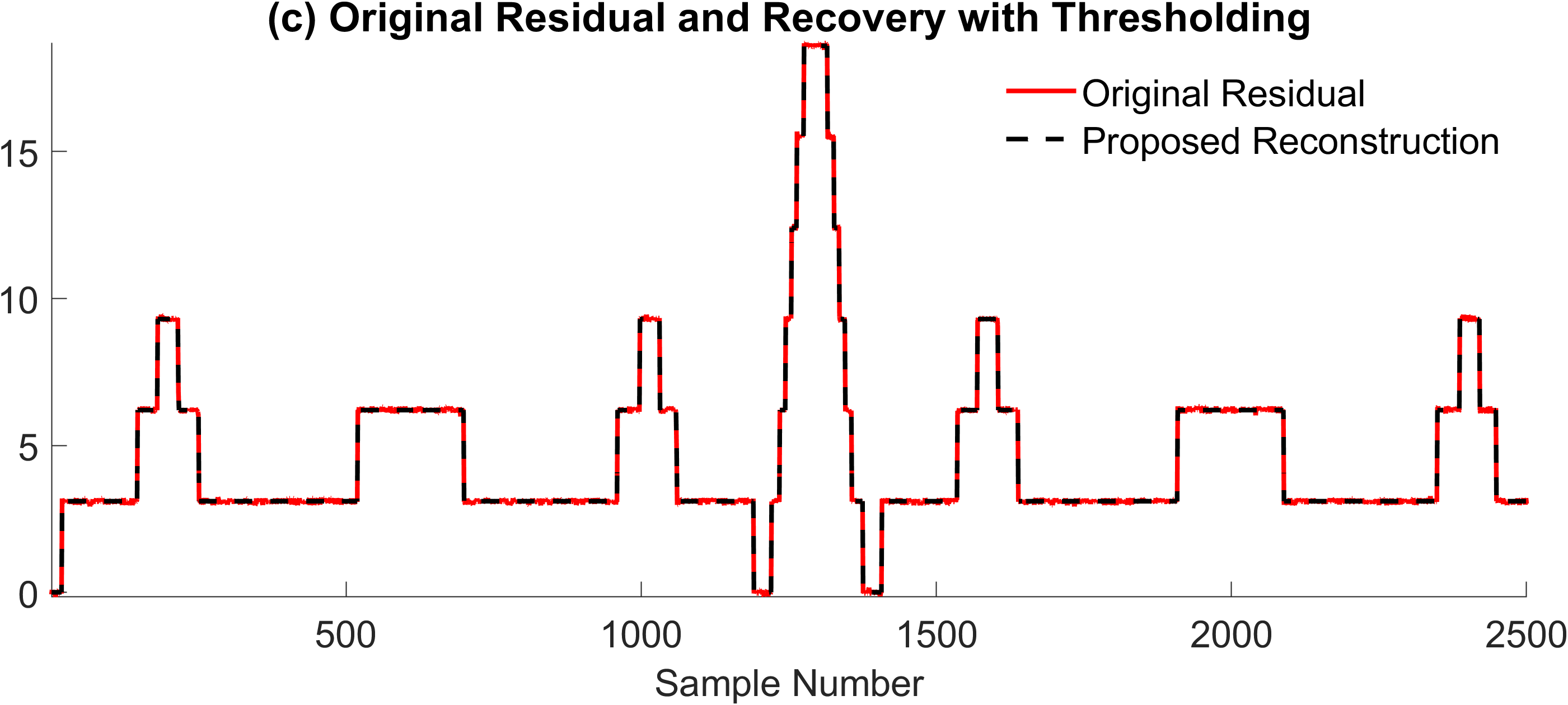}
\label{fig:measured_res2}}
\subfloat{\includegraphics[width=.46\textwidth]{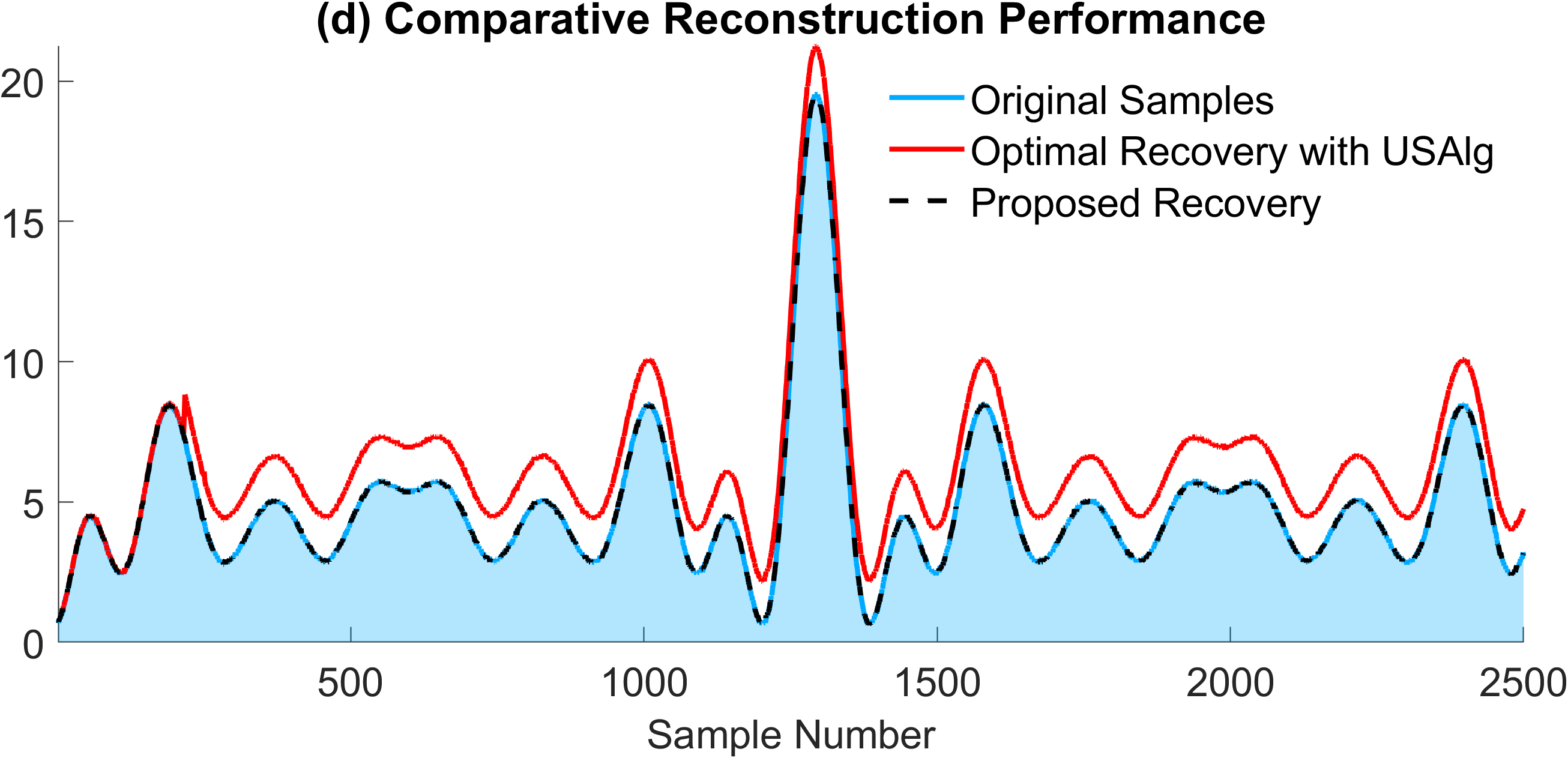}
\label{fig:measured_rec2}}
\caption{Hardware experiment (Bandlimited Signal). (a) Input and output of the proposed encoder. (b) Estimating the folding times by thresholding the filtered data. (c) Reconstruction of the residual function. (d) Input samples and reconstruction with $\mathsf{USAlg}$ with effective threshold and the proposed method. }
\label{fig_measured2}
\end{figure*}

\begin{figure*}[!t]
\centering
\subfloat{\includegraphics[width=.33\textwidth]{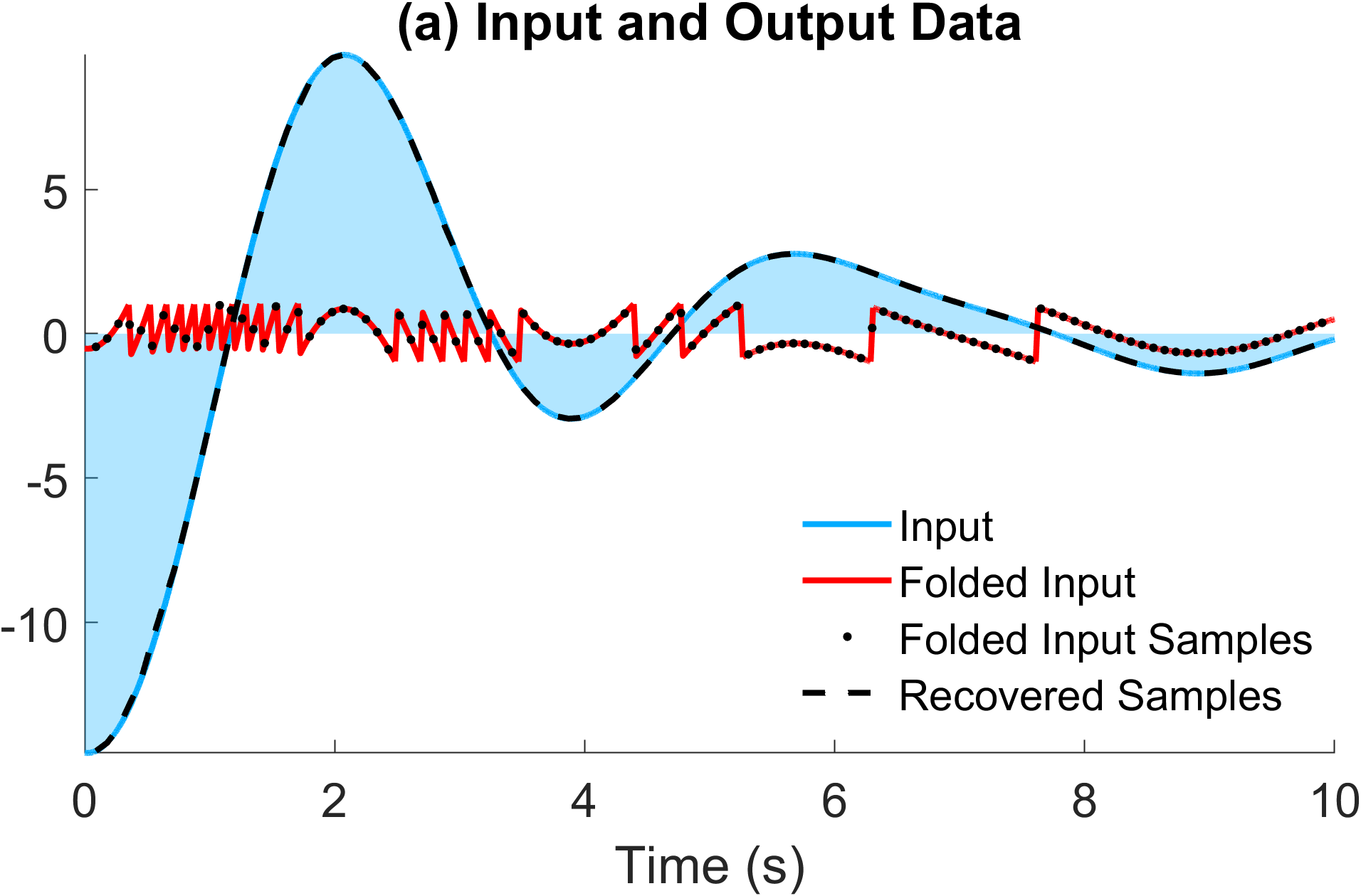}
\label{fig:synth_data}}
\subfloat{\includegraphics[width=.317\textwidth]{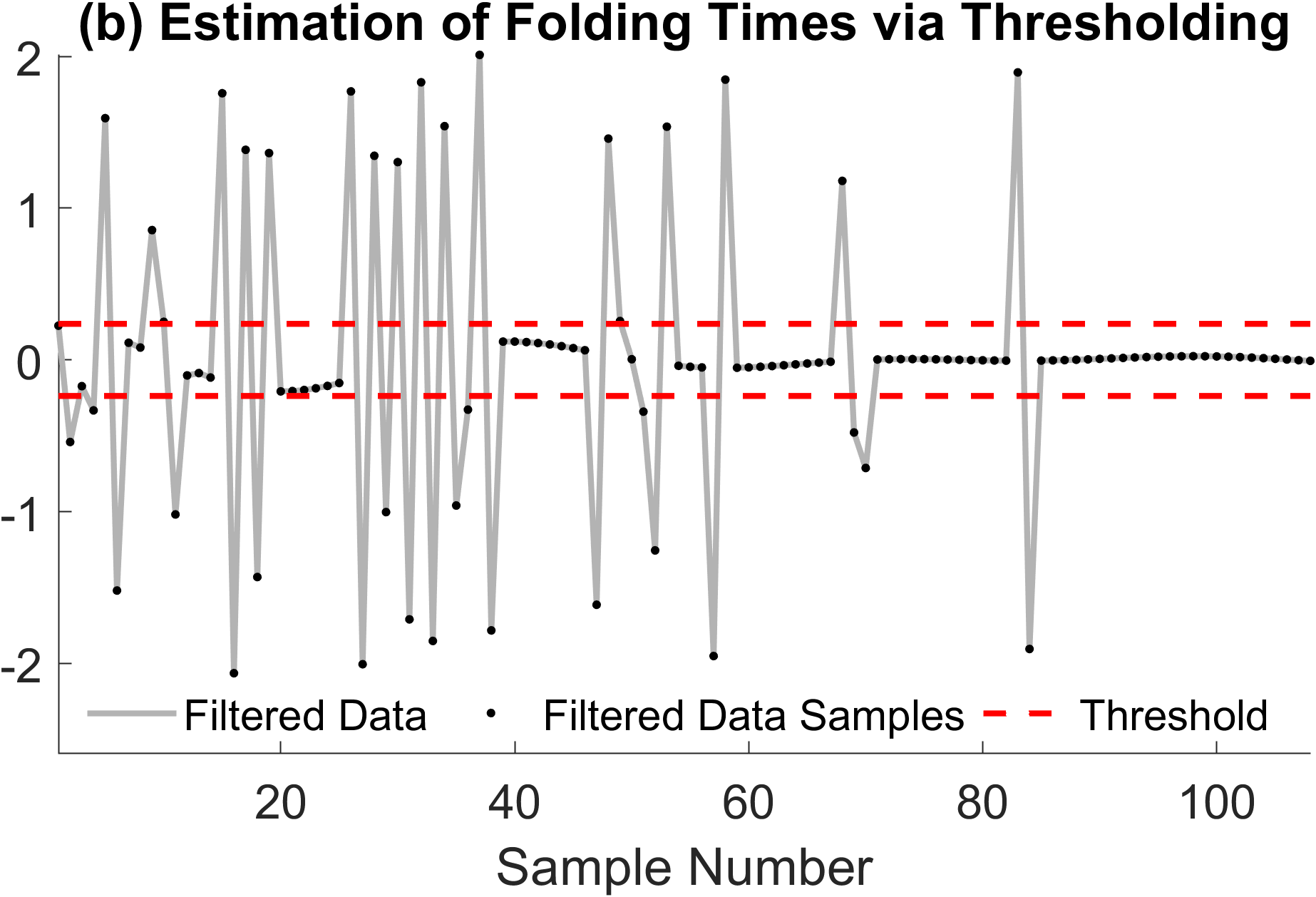}
\label{fig:synth_th}}
\subfloat{\includegraphics[width=.33\textwidth]{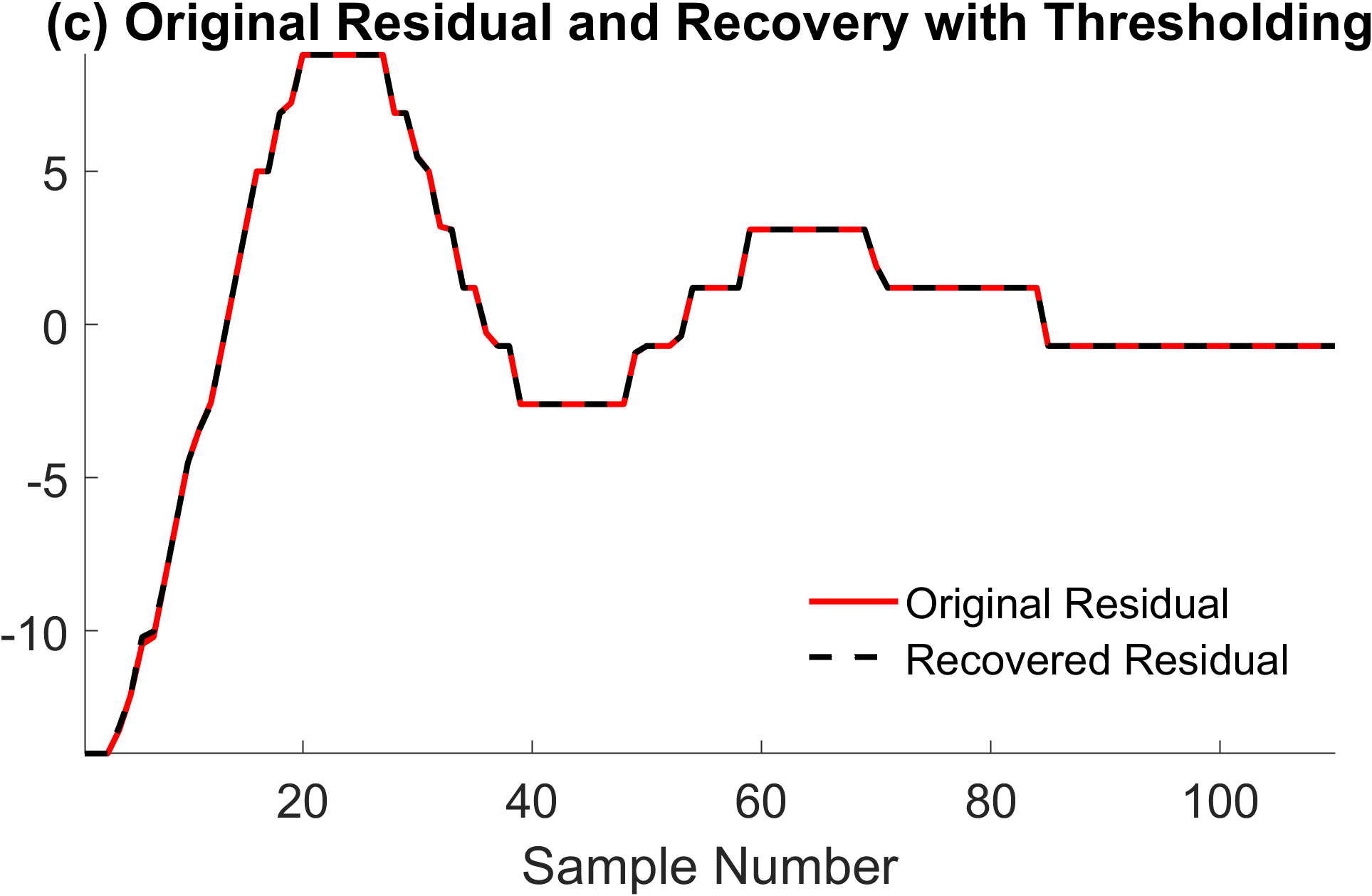}
\label{fig:synth_res}}
\caption{{Simulations with low sampling rates. (a) A randomly generated signal, its modulo samples, and corresponding input reconstruction. (b) Filtered data
and the resulting estimated folding times via thresholding. 
(c) The residual function recovered with the proposed method.}}
\label{fig_sim2}
\end{figure*}

\subsection{Input Recovery at Lower Sampling Rates}

In this experiment, we generated a random non-periodic signal as in Section \ref{sect:experiments_synthetic} with bandwidth $\Omega=2.25\ \mathrm{rad/s}$, centered in $5$ points spaced by $\pi/\Omega$ in $\sqb{0,10}$. The $\mathrm{sinc}$ coefficients were randomly generated using the uniform distribution on $\sqb{-20\lambda,20\lambda}$ with $\lambda=1$. The input is subsequently scaled such that $\norm{g}_\infty=14$. The modulo output, generated with $h=0.1$ and $\alpha=0.02$, is sampled with $T=0.09\ \mathrm{s}$. The input and output are depicted in \fig{fig_sim2}(a). We note that in this experiment condition $\mathrm{TH}_1$ is true, but $\mathrm{TH}_2$ is not satisfied. In fact, the finite difference of the input samples satisfies $\max\vb{\gamma\sqb{k+1}-\gamma\sqb{k}}=1.69\cdot \lambda$. However, recovery with the proposed method is still possible given there is at least one sample between any subsequent folds. 

We first attempt to recover the input using $\mathsf{USAlg}$. For $N=1$, the optimal choice of $\lambda_\mathsf{USAlg}=1.345$ was computed as in Section \ref{sect:experiments_synthetic}, with a corresponding error of $\mathsf{Err}_{\mathsf{USAlg}}=96.8\%$. For $N=2$, the error increased to $\mathsf{Err}_{\mathsf{USAlg}}=300\%$, and subsequently increased exponentially for larger $N$.

We next implement Algorithm \ref{alg:2} for $N=2$ and $\thbar=\lambda_h$, with anchor point $k_0=1$. The anchor point, corresponding to a region of unfolded samples of length $N$, was determined from the thresholded data in \fig{fig_sim2}(b). The reconstructed residual is depicted in \fig{fig_sim2}(c), and subsequently the reconstructed input samples $\widetilde{\gamma}\sqb{k}$ are in \fig{fig_sim2}(a). The recovery error is $\mathsf{Err_{TH}}=0.71\%$. We then increased the dynamic range of the input so that $\norm{g}_\infty=16$, and repeated the same procedure. In this case, we have $\max\vb{\gamma\sqb{k+1}-\gamma\sqb{k}}=1.94\cdot\lambda$, and an error of $\mathsf{Err_{TH}}=1.2\%$.

\section{Proofs}
\label{sect:Proofs}

\bigskip
Here we prove Lemmas 2, 3, Theorem 1 and Proposition 1.

\bigskip
\begin{proof}[Proof of Lemma \ref{lem:supp_psi_eps}]
By the definition of the $n_p$'s, we have that $$\Delta \varepsilon_\gamma =2\lambda_h \sum\nolimits_{p\in\Z} s_p\sqb{\beta_p \delta_{n_p} + \rb{1-\beta_p}\delta_{n_p+1}},$$ where $\beta_p\in\sqb{0,1}$ satisfies
\begin{equation}
\label{eq:beta_p}
\beta_p=\left\lbrace\begin{array}{cc}
\frac{n_p T-\tau_p}{\alpha}, & \quad n_p T<\tau_p+\alpha\\
1, & \quad n_p T\geq\tau_p+\alpha.
\end{array}
\right.
\end{equation}
The result then follows from explicitly calculating finite differences of the Dirac delta function, which amounts to
\begin{equation}
\label{eq:psi_eps}
\psi_N\ast\varepsilon_\gamma\sqb{k}=2\lambda_h s_p \sum\limits_{p\in\Z} \beta_p d^{N-1}_{n_p}+\rb{1-\beta_p}d^{N-1}_{n_p+1},
\end{equation}
where
$d^{N-1}_{n_p}\triangleq\Delta^{N-1}\sqb{k-n_p}$.
\end{proof}

\bigskip
\begin{proof}[Proof of Lemma \ref{lem:k1_k2}]
By \eqref{eq:psi_y} and Lemma \ref{lem:supp_psi_eps}, $\psi_N\ast y_\eta$ and $\psi_N\ast\gamma_\eta$ cannot differ outside of $\mathbb{S}_N$, and hence
\begin{equation}
\label{eq:condI}
\vb{\psi_N\ast y_\eta\sqb{k}}=\vb{\psi_N\ast\gamma_\eta\sqb{k}}<\frac{\lambda_h}{2N}, \forall k \notin \mathbb{S}_N.
\end{equation}
Consequently, given the definition of $\mathbb{M}_N$, we must have $\mathbb{S}_N \supseteq \mathbb{M}_N \supseteq \subSz$. Here, the residual is normalized by defining $e_1\sqb{k}=\frac{\varepsilon_\gamma\sqb{k}}{2\lambda_h s_1}$.
The proof will crucially exploit that, via (\ref{eq:psi_y},\ref{eq:condI}),
\begin{equation}
\label{eq:lower_bound_y}
|\psi_N\ast y_\eta\sqb{k}|
> 2\lambda_h |\rb{\psi_N\ast e_1} [k]| - \frac{\lambda_h}{2N}, \forall k \in \mathbb{Z}.
\end{equation} 
Specifically, we will identify the values $k$ for which $|\rb{\psi_N\ast e_1} [k]|>\frac{1}{2N}$, and deduce from \eqref{eq:lower_bound_y} that these values are all contained in $\mathbb{M}_N$. This yields upper bounds for $k_m$ and lower bounds for $k_M$. Here it makes a noticeable difference if the next sampling time $n_p T$ after $\tau_p$ takes place very close to the edges of the transient or in the middle, or, equivalently, whether the parameter $\beta_1$ is close to the edges of the interval $[0,1]$. More precisely, we consider the following cases.
\begin{enumerate}[leftmargin = *, label = $\alph*)$]
\item   $\beta_1\in \sqb{\frac{1}{2N},1-\frac{1}{2N}} \Leftrightarrow n_1 T \in \mathbb{I}_1$. Then, due to (\ref{eq:psi_eps}),
\end{enumerate}
\begin{equation}
   \label{eq:bound_e12}
\vb{\psi_N\ast e_1\sqb{n_1 +1}}= (1-\beta_1) \vb{\Delta^{N-1}\sqb{0}} \geq\frac{1}{2N}
\end{equation}
and similarly,
\begin{equation}\label{eq:bound_e11} 
\vb{\psi_N\ast e_1\sqb{n_1-\rb{N-1}}}\geq\frac{1}{2N}.
\end{equation}
That is, due to \eqref{eq:lower_bound_y}, $n_1-N+1$ and $n_1+1$ are both contained in $\mathbb{M}_N$. Together with our previous observation that $\subSz \subset \mathbb{S_N}$, we conclude that $k_m=n_1-N+1$, $k_M=n_1+1$, and hence $k_M-k_m=N$, as desired.
\begin{enumerate}[leftmargin = *, label = $\alph*)$, start=2]
\item  $\beta_1\in\sqb{0,1} \setminus\sqb{\frac{1}{2N},1-\frac{1}{2N}} \Leftrightarrow n_1 T \not\in \mathbb{I}_1$.
\end{enumerate}	
\begin{enumerate}[leftmargin = 50pt, label = $b_\arabic*)$]
\item $\beta_1\in\left[0,\frac{1}{2N}\right)$. In this case, \eqref{eq:bound_e12} and the facts that $k_M=n_1+1$ and $k_m\geq n_1-N+1$ can be established as before, but
$\vb{\psi_N\ast e_1\sqb{n_1-\rb{N-1}}}$ can be arbitrarily small. However, the following is satisfied.	
\begin{align}
\nonumber \vb{\psi_N\ast e_1\sqb{n_1-N+2}} & =\vb{\beta_1\Delta^{N-1}\sqb{-N+2}+\rb{1-\beta_1}\Delta^{N-1}\sqb{-\rb{N-1}}}\\
&= \vb{1-\beta_1 N}\geq\frac{1}{2}\geq \frac{1}{2N}.
\label{eq:bound_e13}
\end{align}
Thus $k_m\leq n_1-N+2$, which entails all the desired conditions via \eqref{eq:lower_bound_y}.

\item  $\beta_1\in\left(1-\frac{1}{2N},1\right]$. Here \eqref{eq:bound_e11} and the observations that $k_m=n_1-N+1$ and $k_M\leq n_1+1$ hold as in $a)$, but $\vb{\psi_N\ast e_1\sqb{n_1 +1}}$ cannot be bounded below. In direct analogy to \eqref{eq:bound_e13}, we instead deduce that
\begin{align}
\label{eq:bound_e16}
\vb{\psi_N\ast e_1\sqb{n_1}}&\geq \frac{1}{2} \geq \frac{1}{2N}.
\end{align}
\item $\beta_1=1$. Here we have that $n_p T\geq \tau_p +\alpha$ and due to (\ref{eq:psi_eps}) it follows that $\mathrm{supp} \rb{\psi_N\ast \varepsilon_\gamma}=\cb{n_1-\rb{N-1},\dots,n_1}$, and, due to \eqref{eq:condI}, it follows that $k_m=n_1-\rb{N-1}$ and $k_M=n_1$.
\end{enumerate}
\hskip7pt Thus $k_M\geq n_1$, which again yields all of the desired
 conditions via \eqref{eq:lower_bound_y}.
\end{proof}

\bigskip
\begin{proof}[Proof of Theorem \ref{theo:s1_tau1}]
To show that $\widetilde{s}_1=s_1$, we observe that, 
$$\vb{\psi_N\ast\gamma_\eta\sqb{k_m}}<\frac{\lambda_h}{2N} \quad \mbox{and} \quad \vb{\psi_N\ast y_\eta\sqb{k_m}}>\frac{\lambda_h}{2N},$$ 
which implies $\mathrm{sign} \rb{\psi_N\ast y_\eta \sqb{k_m}}=-\mathrm{sign} \rb{\psi_N\ast \varepsilon_\gamma \sqb{k_m}}$, and due to (\ref{eq:psi_eps})
\begin{equation}
\label{eq:sign}
\begin{split}
\mathrm{sign} \rb{\psi_N\ast y_\eta \sqb{k_m}}&=-s_1 . \mathrm{sign} \left(\beta_1\Delta^{N-1}\sqb{k_m-n_1}\right.+ \left.\rb{1-\beta_1}\Delta^{N-1}\sqb{k_m-n_1-1}\right).
\end{split}
\end{equation}
Thus we can compute $s_1$ by dividing the left hand side by the second factor on the right-hand side of \eqref{eq:sign}, which we derive in the following for each particular case.
\begin{enumerate}[label=\alph*)]
    \item $k_M-k_m=\const{N}$. Here the facts that $k_m=n_1-\rb{N-1}$ and $\widetilde{n}_1=n_1$ are direct byproducts of the proof of Lemma \ref{lem:k1_k2}. 
Consequently, $\Delta^{N-1}\sqb{k_m-n_1-1}=0$ and $s_1=\widetilde{s}_1$ follows from \eqref{eq:sign} given that $\Delta^{N-1}\sqb{-N+1}=1$. 
By (\ref{eq:mid_transient}-\ref{eq:after_transient}) the first case of the definition of $\beta_1$ in \eqref{eq:beta_p} takes effect and one has that $\tau_1= n_1 T+\alpha \beta_1$. Hence,
\begin{equation*}
    |\widetilde{\tau}_1 -\tau_1| = \ \alpha\left|\beta_1-\frac{(-1)^{N-1}\psi_N\ast y_\eta\sqb{n_1}+2\lambda_h s_1}{2\lambda_h s_1}\right|.
\end{equation*}
The expression of $\beta_1$ is derived from (\ref{eq:psi_y},\ref{eq:psi_eps}), which yields 
\begin{equation*}
\vb{\widetilde{\tau}_1-\tau_1} =  \frac{\alpha\vb{\psi_N\ast \gamma_\eta\sqb{n_1}}}{2\lambda_h  N},
\end{equation*}
which according to \eqref{eq:condIa} proves \eqref{eq:tau_p_bound}.
\item $k_M-k_m=\const{N}-1$. Similarly to part a), {Equations} \eqref{eq:bound_start_transient} and \eqref{eq:bound_end_transient} are direct byproducts of the proof of Lemma~\ref{lem:k1_k2}.
To show that $\widetilde{s}_1 =s_1$, we recall from Lemma~\ref{lem:k1_k2} that one has either $k_m=n_1-N+1 \Rightarrow \beta_1\in\sqb{0,1}$ or $k_m=n_1-N+2 \Rightarrow \beta_1\in \left[0,\frac{1}{2N}\right)$. We evaluate \eqref{eq:sign} for each case, where $\Delta^{N-1}\sqb{-N}=0$, $\Delta^{N-1}\sqb{-N+1}=1$, and $\Delta^{N-1}\sqb{-N+2}=-\rb{N-1}$, which yields \eqref{eq:s_1}.
\end{enumerate}
\end{proof}

\bigskip
\begin{proof}[Proof of Proposition \ref{prop:mse}]
Given $K$ samples, the error satisfies
\begin{align}
\mathsf{MSE}\rb{\widetilde{\gamma},\gamma}
&= \frac{1}{K}\sum\limits_{k = 1}^K {{{\left( {{{\widetilde \varepsilon }_\gamma }\left[ k \right] - {\varepsilon _\gamma }\left[ k \right]} \right)}^2}}
 = \frac{1}{K}\sum\limits_{k = 1}^K {{{\left( {\sum\limits_{p = 1}^P {s_p}\ETP{kT} } \right)}^2}}
 \label{eq:exprMSE}
\end{align}
where $\ETP{t}= \varepsilon _0^{{{\widetilde \tau }_p}}\left( {t} \right) - \varepsilon _0^{{\tau _p}}\left( {t} \right)$ and $\varepsilon_0^{\tau^*}\rb{t}=\varepsilon_0\rb{t-{\tau}^*}$, $\forall\tau^*\in\mathbb{R}$, so $E_p(t)=0$ for all $t\notin \mathbb{A}$, where $$\mathbb{A}=\left(\min \cb{\tau_p,\widetilde{\tau}_p},\max \cb{\tau_p+\alpha,\widetilde{\tau}_p+\alpha}\right)$$ is a set including both transient regions.
As Theorem \ref{theo:s1_tau1} yields that the two elements in each pair $\rb{\widetilde{\tau}_p, \tau_p}$ are relatively close, this implies that many of the terms in \eqref{eq:exprMSE} are $0$. 
Specifically, we will show that $\ETP{kT} = 0, \forall k \neq n_p$. 
There are two cases
\begin{enumerate}[leftmargin = *, label = $\alph*)$]
\item $k_M^p-k_m^p=N$. 
Then due to Theorem \ref{theo:s1_tau1}, we have that 
$\widetilde{\tau}_p \in \sqb{\tau_p-\frac{\alpha}{4N^2},\tau_p+\frac{\alpha}{4N^2}}$. Observe that it follows from (\ref{eq:mid_transient},\ref{eq:start_end_transient}) that $n_p T\in [\tau_p, \tau_p+\alpha[ $. Furthermore, $$T\geq \alpha +\frac{\alpha}{4N^2}, \quad (n_p-1)T \leq \tau_p-\frac{\alpha}{4N^2} \quad \mbox{and} \quad (n_p+1)T \geq \tau_p+\alpha+\tfrac{\alpha}{4N^2},$$ so $n_pT$ is the only sampling time in $\mathbb{A}$.
It follows that
\end{enumerate}
\begin{align}
\rb{\ETP{n_pT}}^2 &=\rb{2 \lambda_h \sqb{\frac{n_p T-\widetilde{\tau}_p}{\alpha}-\frac{n_p T-{\tau}_p}{\alpha}}}^2
=4\lambda_h^2 \frac{\rb{\widetilde{\tau}_p-\tau_p}^2}{\alpha^2}\leq  \frac{\lambda_h^2}{4N^4}.
\label{eq:bound_errora}
\end{align}
\begin{enumerate}[leftmargin = *, label = $\alph*)$, start=2]
\item $k_M^p-k_m^p=N-1 \Rightarrow \widetilde{\tau}_p=\widetilde{n}_pT$. We have two cases
\begin{enumerate}[leftmargin = 50pt, label = $b_\arabic*)$]
\label{item:b1}
\item $n_p T \in \left[\tau_p, \tau_p+\frac{\alpha}{2N}\right) \Rightarrow \widetilde{\tau}_p=n_p T$, which directly implies that $kT\notin \mathbb A$ for $k\neq n_1$.
Furthermore, in analogy to \eqref{eq:bound_errora}, we obtain that

\begin{align}
{\rb{\ETP{n_p T}}^2}
& =\rb{2 \lambda_h \sqb{0-\frac{n_p T-{\tau}_p}{\alpha}}}^2 \leq \frac{\lambda_h^2}{N^2}.
\label{eq:bound_errorb}
\end{align}
\item $n_p T \in \left( \tau_p + \alpha -\frac{\alpha}{2N}, \tau_p+T \right)$ $ \Rightarrow \widetilde{\tau}_p=\rb{n_p-1}T\in \left( \tau_p+\alpha - \tfrac{\alpha}{2N}-T, \tau_p\right)$, which again implies that 
$n_pT$ is the only sampling time in $\mathbb{A}$. Noting that the range of $\widetilde\tau_p$ implies that $\varepsilon_0^{\widetilde \tau_p} = 2\lambda_h$, we observe that
\begin{equation}
    \rb{\ETP{n_p T}}^2=\rb{2\lambda_h-\varepsilon_0^{\tau_p}\rb{n_p T}}^2.
\end{equation}
Now $2\lambda_h\geq \varepsilon_0^{\tau_p}\rb{n_p T} \geq \varepsilon_0^{\tau_p}\rb{\tau_p+\alpha-\frac{\alpha}{2N}}=\varepsilon_0\rb{\alpha-\frac{\alpha}{2N}}=2\lambda_h \rb{1-\frac{1}{2N}}$ yields that, $$2\lambda_h-\varepsilon_0^{\tau_p}\rb{n_p T} \in \sqb{0,\frac{\lambda_h}{N}} \Rightarrow \rb{\ETP{n_p T}}^2\leq \frac{\lambda_h^2}{N^2}.$$
\end{enumerate}
\end{enumerate}

Then, the only potentially non-zero terms in \eqref{eq:exprMSE} are for $k=n_{p^*}$ and $p=p^*$ for $p^*=1,\dots,P$, which amounts to
\begin{align*}
\mathsf{MSE}\rb{\widetilde{\gamma},\gamma} &= \frac{1}{K} \sum\limits_{p^*=1}^P \rb{ s_{p^*} \ETPP{p^*}{n_{p^*}T}^2}\\
&\leq \frac{P}{K} \max \ETPP{p^*}{n_{p^*}T}^2\leq \frac{P}{K} \cdot \frac{\lambda_h^2}{N^2}.
\end{align*}
The value of $P$ can be further bounded if we consider that the minimum distance between two analog folding times satisfies \eqref{eq:min_dist_analog_folds}. Therefore, $P$ is bounded by $P\leq \rb{{KT\Omega g_\infty}/{h^*}}+1$ which completes the proof.
\end{proof}

\section{Conclusion}
\label{sect:conclusion}

 {
In this paper, we introduced a \emph{computational sensing} approach for single-shot high dynamic range (HDR) signal reconstruction, which allows to compensate imprecise hardware by sophisticated algorithms that can account for the model mismatch in hardware implementations. 
To this end we introduce an end-to-end pipeline consisting of a generalized modulo encoding model, a theoretically guaranteed reconstruction method and an associated hardware prototype.
The pipeline leverages the concept of hysteresis, caused by the misalignment of the reset threshold and post-reset value, in several unique ways. In modeling, hysteresis allows the representation of real data generated by circuits with reduced calibration. Instead of being detrimental on the recovery front, surprisingly, hysteresis enables introducing recovery guarantees for the generalized model in the noiseless and noisy scenarios. A hardware prototype with off-the-shelf components demonstrates the effect of hysteresis and is used as a testbed for experimental demonstration.
Specifically, in experiments we report reconstruction of signals that are {$28$ times} larger than the ADC threshold. Owing to the co-design approach that involves both hardware and algorithms, our work raises a number of interesting questions for future discussion. 
}
\subsection{Future Work}

\begin{enumerate}
\item A direct next step in our research is to consider wider classes of signals (beyond the bandlimited input model). This motivates the study of optimal filter design and new reconstruction algorithms specifically tailored for particular signal classes. 

\item The methodology is currently assuming that the sampling period is larger than the folding transient time $T\geq \alpha$. Transient times observed experimentally are very small, so this is not a restriction in practice. However, the benefits outlined in this paper motivates the design and analysis of artificially longer transient times to further extend the proposed methodology. Having several samples in the transient period would likely provide more accurate reconstructions.
\item Currently our assumption is that there is a sampling time between any two subsequent folds. However, given that the number of samples between folds is quite variable, this can be extended to the case where the number of folds is at most equal to the number of samples, leading to further reduced recovery conditions.
\end{enumerate}

\ifCLASSOPTIONcaptionsoff
\newpage
\fi

\bibliographystyle{IEEEtran}
\bibliography{IEEEabrv.bib,Ref2020new.bib}

\clearpage

\end{document}